\documentclass{article}
\pdfoutput=1

\usepackage{amsmath}
\usepackage{amssymb}
\usepackage{color}
\usepackage{url}

\usepackage{algorithm}
\usepackage[noend]{algorithmic} 
\usepackage{amsmath}
\usepackage{amsthm}
 \newtheorem{theorem}{Theorem}
 \newtheorem{proposition}{Proposition}

 \newtheorem{corollary}{Corollary}
 
\newtheorem{lemma}{Lemma}
\usepackage{subfig}

\usepackage{graphicx}

\graphicspath{{../}{../../results/001/}{../../results/002/}{../../results/003/}{../../../svn/lemurpython/timeseries/}}

\usepackage{setspace}

\begin{document}




\title{Faster Sequential Search with a Two-Pass  Dynamic-Time-Warping  Lower Bound}



 \author{Daniel Lemire\footnote{
 LICEF, Universit\'e du Qu\'ebec \`a Montr\'eal (UQAM), 100 Sherbrooke West, Montreal~(Quebec), H2X 3P2 Canada}}

\date{}
\maketitle
\begin{abstract}
The Dynamic Time Warping (DTW) is a popular similarity measure between time series.
The DTW fails to satisfy the triangle inequality and its computation 
requires quadratic time.
Hence, to find  closest neighbors quickly, we use bounding 
techniques.
We can avoid most DTW computations with an inexpensive lower bound (LB\_Keogh).
We compare LB\_Keogh with a tighter lower bound (LB\_Improved). 
We find that LB\_Improved-based search is faster for sequential search.
 As an example, 
our approach is 3~times faster over random-walk and shape time series.
We also review some of the mathematical properties of the DTW\@. We derive
a tight triangle inequality for the DTW. We show that the DTW becomes the $l_1$ distance
when time series are separated by a constant. 
\end{abstract}



\section{Introduction}

Dynamic Time Warping (DTW) was initially introduced
to recognize spoken words~\cite{sakoe1978dpa}, but
it has since been applied to a wide range of information
retrieval and database problems:
handwriting recognition~\cite{bahlmann2004wio,niels2005udt},
 signature recognition~\cite{1219544,chang2007mdt}, image de-interlacing~\cite{almog2005},
  appearance matching for security purposes~\cite{Kale2004}, 
  whale vocalization classification~\cite{brown2006cvk}, 
  query by humming~\cite{4432642,zhu2003wie},
  classification of motor activities~\cite{muscillo2007cma},
 face localization~\cite{lopez2007flc},
 chromosome classification~\cite{legrand2007ccu},
  shape retrieval~\cite{TPAMI.2005.21,marzal2006cbs}, and so on. 
Unlike the Euclidean distance, DTW optimally aligns or ``warps'' 
the data points of
two time series (see Fig.~\ref{fig:example}). 

When the distance between two time series forms a metric, such as
the Euclidean distance or the Hamming distance,
several indexing or search techniques have been 
proposed~\cite{502808,1221301,958948,230528,1221193}.
However, even assuming that we have a metric, Weber et al. have shown
that the performance of any indexing scheme degrades  to that of a sequential scan,
when there are more than a few dimensions~\cite{671192}.
Otherwise---when the distance is not a metric or that
the number of dimensions is too large---we use 
bounding techniques such as the
Generic multimedia object indexing (GEMINI)~\cite{faloutsos1996smd}.
We quickly discard (most) false positives
by computing a lower bound. 

\begin{figure}[hb]
\centering\includegraphics[width=0.5\columnwidth]{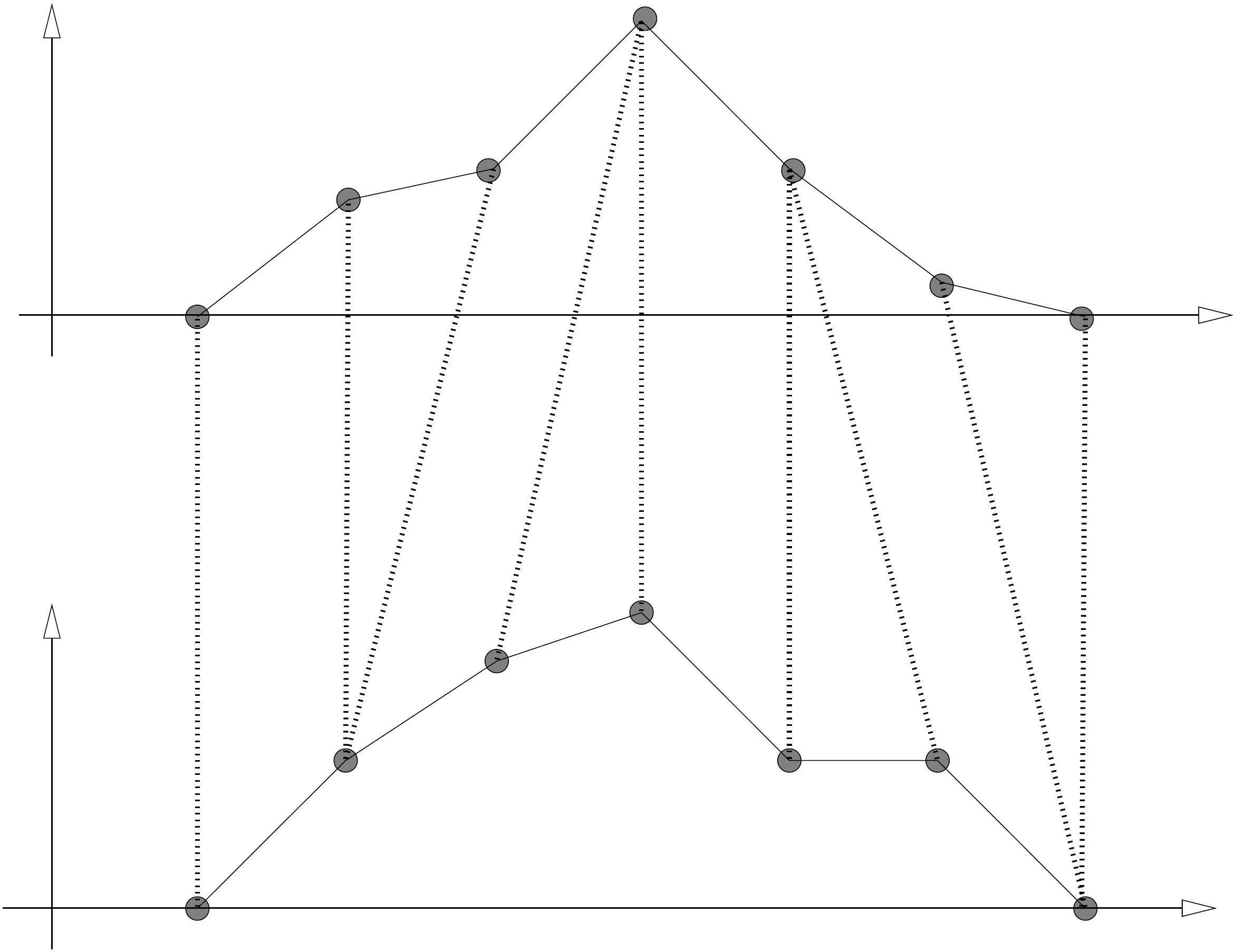}
\caption{\label{fig:example}Dynamic Time Warping example}
\end{figure}


Ratanamahatana and Keogh~\cite{ratanamahatana2005tmd} argue that their lower bound
(LB\_Keogh) cannot be improved upon. To make their point, they report that LB\_Keogh
allows them to prune out over 90\% of all
DTW computations on several data sets. 

We are able to improve upon LB\_Keogh as follows.
The first step of our two-pass approach is LB\_Keogh itself.
If this first lower bound is sufficient to discard the candidate, then
the computation terminates and the next candidate is considered.
Otherwise, we process the time series a second time to increase the
lower bound. If this second lower bound is large enough,
the candidate is pruned, otherwise we compute the full DTW\@. 
We show experimentally
that the two-pass approach can be several times faster.

The paper is organized as follows. In Section~\ref{sec:dtw}, we 
 define the DTW in a generic manner
as the minimization of the $l_p$ norm ($\text{DTW}_p$).
In Section~\ref{sec:properties}, we present various secondary mathematical results.
Among other things, we show that if $x$ and $y$ are separated by a constant ($x\geq c \geq y$ or $x\leq c \leq y$)
then the $\text{DTW}_1$ is the $l_1$ norm (see Proposition~\ref{prof:band1}).
In Section~\ref{sec:ti}, we derive a tight triangle inequality for the DTW.
In Section~\ref{sec:whichmeasure}, we show that $\text{DTW}_{1}$ is good choice
for time-series classification.
In Section~\ref{sec:definingUL}, we compute generic lower bounds on the DTW and their
approximation errors using warping envelopes.
In Section~\ref{sec:envelopes}, we show how to compute the warping
envelopes quickly and derive some of their mathematical properties.
The next two sections introduce LB\_Keogh and LB\_Improved respectively,
whereas the last section presents an experimental comparison.

\section{Conventions}

Time series are arrays of values measured at certain times.
For simplicity, we assume a regular
 sampling rate so that time series are generic arrays of floating-point values.
A time series $x$ has length $\vert x\vert $.
Time series have length $n$ and
are indexed from 1 to $n$.
The $l_p$ norm of  $x$ 
is $\Vert x \Vert_p= (\sum_i \vert x_i\vert ^p)^{1/p}$ for any integer
$0<p<\infty$ and $\Vert x \Vert_\infty= \max_i\vert  x_i\vert$. 
The $l_p$ distance between $x$ and $y$ is  $\Vert x - y\Vert_p$
and it satisfies
the triangle inequality $\Vert x - z\Vert_p\leq \Vert x - y\Vert_p+ \Vert y - z\Vert_p$ for $1\leq p\leq \infty$.
Other conventions are summarized in Table~\ref{tab:commonnot}.

\begin{table}[bth]
\caption{\label{tab:commonnot}Frequently used conventions}
\begin{tabular}{ccc}\hline
$\vert x \vert$ or $n$ & length  \\
$\Vert x \Vert_p$ & $l_p$ norm \\
$\text{DTW}_{p}$ & monotonic DTW \\
$\text{NDTW}_{p}$ & non-monotonic DTW \\
$w$ & DTW locality constraint\\ 
$N(\mu,\sigma)$ & normal distribution \\
$U(x), L(x)$ & warping envelope (see Section~\ref{sec:definingUL})\\
$H(x,y)$ & projection of $x$ on $y$ (see Equation~\ref{eqn:hxy})\\
\hline
\end{tabular}
\end{table}
\section{Related Works}

Beside DTW, several similarity metrics have been
proposed including the
directed and general  Hausdorff distance, Pearson's correlation,
 nonlinear elastic matching distance~\cite{veltkamp2001sms}, 
 Edit distance with Real Penalty (ERP)~\cite{chen:mln},
Needleman-Wunsch similarity~\cite{needleman1970gma},
Smith-Waterman similarity~\cite{smith1981icm},
and SimilB~\cite{1340465}.

Dimensionality reduction, such as piecewise constant~\cite{keogh2005eid} or piecewise 
linear~\cite{xiao2004nst,shou2005,dong2006lsd} segmentation,
 can speed up retrieval under DTW distance.
 These techniques can be coupled with other optimization techniques~\cite{Sakurai2005}.

The performance of lower bounds can be further improved if one 
uses  early abandoning~\cite{1106385} to cancel the computation of the
lower bound as soon as the error is too large. 
Boundary-based lower-bound functions sometimes outperform LB\_Keogh~\cite{zhou2007bbl}.
Zhu and Shasha showed that computing a warping envelope 
 prior to applying
dimensionality reduction results in  a tighter lower bound~\cite{zhu2003wie}.
We can also quantize~\cite{1166502} or cluster~\cite{keogh2006lse} the time series.

\section{Dynamic Time Warping}
\label{sec:dtw}
A many-to-many matching between the data points in time series $x$ and the data point
in time series $y$ matches every data point $x_i$ in $x$ with at least one data point $y_j$ in $y$,
and every data point in $y$ with at least a data point in $x$. The set of 
matches $(i,j)$ forms a \emph{warping path} $\Gamma$.
We define the DTW as the minimization of the $l_p$ norm of the differences $\{ x_i-y_j  \}_{(i,j)\in \Gamma}$
over all warping paths. 
A warping path is minimal if there is no subset $\Gamma'$ of $\Gamma$ 
forming a warping path: for simplicity we require all warping paths to be minimal.

In computing the DTW distance, we commonly require the warping to remain
local. For time series $x$ and $y$, we do not align values 
$x_i$ and $y_j$ if $\vert i- j \vert > w$ for some locality constraint $w\geq 0 $~\cite{sakoe1978dpa}.
When $w=0$, the DTW becomes the $l_p$ distance whereas when $w\geq n$, the DTW has no locality constraint.
The value of the DTW diminishes monotonically as $w$ increases.
 
 Other than locality, DTW can be monotonic: if we align value $x_i$
 with value $y_j$, then we cannot align value $x_{i+1}$ with a value
 appearing before $y_j$ ($y_{j'}$ for $j'<j$).

We note the DTW distance between $x$ and $y$ using the $l_p$ norm as $\text{DTW}_p(x,y)$
when it is monotonic and as $\text{NDTW}_p(x,y)$ when monotonicity is not required.

By dynamic programming, the monotonic DTW  requires $O(w n)$ time. A typical value of $w$ is $n/10$~\cite{ratanamahatana2005tmd} 
so that the DTW is in~$O(n^2)$.
To compute the DTW, we use the following recursive formula.
Given an array $x$, we write the suffix starting at position $i$, $x_{(i)}= x_i,x_{i+1},\ldots,x_n$.
The symbol $\oplus$ is the exclusive or.  Write $q_{i,j}= \text{DTW}_p(x_{(i)},y_{(j)})^p$ so that $\text{DTW}_p(x,y)=\sqrt[p]{q_{1,1}}$,
then
\begin{align*}
q_{i,j}= \begin{cases}0  & \text{if $\vert x_{(i)}\vert = \vert (y_{(j)}\vert =0$}\\
\infty  & \begin{matrix}\text{if $\vert x_{(i)}\vert = 0 \oplus \vert y_{(j)}\vert =0$}\\
\text{or  $\vert i- j \vert > w$ }\end{matrix}\\
 \begin{matrix}\vert x_i-y_j \vert^p +\\ \min (
 q_{i+1,j}, q_{i,j+1}, q_{i+1,j+1})\end{matrix} & \text{otherwise.}
\end{cases}
\end{align*}
For $p=\infty$, we rewrite
the preceding recursive formula with 
$q_{i,j}= \text{DTW}_\infty(x_{(i)},y_{(j)})$, and $q_{i,j}=\max (\vert x_i-y_j \vert, \min (
 q_{i+1,j}, q_{i,j+1}, q_{i+1,j+1}))$ when $\vert x_{(i)}\vert \neq 0$,  $\vert y_{(j)}\vert \neq 0$, and $\vert i-j\vert \leq w$.

We can compute $\text{NDTW}_1$ 
without time constraint in $O(n \log n)$~\cite{1275562}: if the values of the
time series are already sorted, the computation is in $O(n)$ time.

We can express the solution of the DTW problem  
 as an alignment of the two initial time series (such as $x=0,1,1,0$ and $y=0,1,0,0$) where some of the 
values are repeated (such as $x'=0,1,1,0,\textbf{0}$ and $y'=0,1,\textbf{1},0,0$).
If we allow non-monotonicity (NDTW), then values can also be inverted.

The non-monotonic DTW is no larger than the monotonic DTW which is itself no 
larger than the $l_p$ norm: $\text{NDTW}_p(x,y) \leq \text{DTW}_p(x,y) \leq \Vert x-y\Vert_p$ for all $0<p\leq \infty$.


\section{Some Properties of Dynamic Time Warping}
\label{sec:properties}

The DTW distance can be counterintuitive. 
As an example, if $x,y,z$ are three time series such that
 $x \leq y \leq z$ pointwise,
then it does not follow that $\text{DTW}_p(x,z) \geq \text{DTW}_p(z,y)$.
Indeed, choose $x=7,0,1,0$, $y=7,0,5,0$, and $z=7,7,7,0$,
then $\text{DTW}_\infty(z,y)=5$ and  $\text{DTW}_\infty(z,x)=1$.
Hence, we review some of the mathematical properties
of the DTW\@. 

The warping path aligns $x_i$ from time series $x$ and $y_j$ from time series $y$ if
$(i,j) \in \Gamma$. The next proposition is a general constraint on warping paths.

\begin{proposition}\label{prop:warpingpathprop} Consider any two time series $x$ and $y$. 
For any minimal warping path,  if $x_i$ is aligned with $y_j$, then either $x_i$ is aligned only with $y_j$
or $y_j$ is aligned only with $x_i$.
\end{proposition}
\begin{proof} Suppose that the result is not true. Then there is $x_k, x_i$ and $y_l,y_j$ such
 that $x_k$ and $x_i$ are aligned with $y_j$, and $y_l$ and $y_j$ are aligned with $x_i$. We can delete $(k,j)$ from the warping path and  still have a warping path. A contradiction.
\end{proof}

Hence, we have that the cardinality of the warping path  is no larger than $2n$.  Indeed,
each match $(i,j)\in \Gamma$ must be such that $i$ or $j$ only occurs in this match
by the above proposition.

The next lemma shows that the DTW becomes the $l_p$ distance when either $x$ or $y$ is constant.

\begin{lemma}\label{lemma:yconstant}
For any $0 < p \leq \infty$, if $y=c$ is a constant, then $\text{NDTW}_p(x,y) = \text{DTW}_p(x,y)= \Vert x-y\Vert_p$.
\end{lemma}

When $p=\infty$, a stronger result is true: if $y=x+c$ for some constant $c$,
then $\text{NDTW}_\infty(x,y) = \text{DTW}_\infty(x,y)= \Vert x-y\Vert_\infty$. Indeed, 
 $\text{NDTW}_\infty(x,y)\geq \vert \max(y)-\max(x) \vert = c = \Vert x-y\Vert_\infty\geq \Vert x-y\Vert_\infty$ which
 shows the result.
This same result is not true for $p<\infty$: for $x=0,1,2$ and $y=1,2,3$, we have
$\Vert x-y\Vert_p=\sqrt[p]{3}$ whereas $\text{DTW}_p(x,y)=\sqrt[p]{2}$.
 However, the DTW is translation invariant:  
$\textrm{DTW}_p(x,z)=\textrm{DTW}_p(x+b,z+b)$
and $\textrm{NDTW}_p(x,z)=\textrm{NDTW}_p(x+b,z+b)$ for any scalar $b$ and $0<p\leq \infty$.

The $\text{DTW}_1$ has the property that if 
the time series are value-separated, then the DTW is the $l_1$ norm
as the next proposition shows.

\begin{proposition}\label{prof:band1}If $x$ and $y$
are such that  either $x\geq c \geq y$ or $x\leq c \leq y$ for some constant $c$,
then $\text{DTW}_1(x,y)=\text{NDTW}_1(x,y)= \Vert x-y \Vert_1$.
\end{proposition}
\begin{proof}Assume $x\geq c \geq y$, there exists $x', y'$ such that
$x'\geq c \geq y'$ and $\text{NDTW}_1(x,y)= \Vert x'-y' \Vert_1 =\sum_i \vert x'_i - y'_i\vert
= \sum_i \vert x'_i - c\vert  + \vert c- y'_i\vert = \Vert x'-c \Vert_1 + \Vert c-y' \Vert_1 \geq
\Vert x-c \Vert_1 + \Vert c-y \Vert_1 =  \Vert x-y \Vert_1$. Since 
we also have  $\text{NDTW}_1(x,y) \leq \text{DTW}_1(x,y) \leq \Vert x-y\Vert_1$, the equality follows.
\end{proof}

Proposition~\ref{prof:band1} does not hold for $p>1$: $\text{DTW}_2((0,0,1,0), (2,3,2,2))=\sqrt{17}$
whereas $\Vert(0,0,1,0) - (2,3,2,2)\Vert_2=\sqrt{18}$.

In classical analysis, we have that $n^{1/p-1/q} \Vert x\Vert_q \geq \Vert x \Vert_p$~\cite{folland84} for  $1\leq p < q \leq \infty$. A similar results is true for the DTW and it allows us to conclude that $\text{DTW}_p(x,y)$ and $\text{NDTW}_p(x,y)$ decrease monotonically as $p$ increases.

\begin{proposition} For $1\leq p < q \leq \infty$, we have that 
$(2n)^{1/p-1/q}  \text{DTW}_q(x,y) \geq \text{DTW}_p(x,y)$ 
where $\vert x\vert = \vert y \vert =n$. The result also holds for the non-monotonic DTW.
\end{proposition}
\begin{proof}The argument is the same for the monotonic or non-monotonic DTW\@.
 Given $x,y$ consider the two aligned (and extended) time series $x', y'$
 such that  $\text{DTW}_q(x,y)=\Vert x'-y' \Vert_q$. As a consequence of  Proposition~\ref{prop:warpingpathprop},
 we have $\vert x'\vert =\vert y'\vert \leq 2n$. From classical analysis, we
 have $\vert x'\vert ^{1/p-1/q} \Vert x'-y'\Vert_q \geq \Vert  x'-y' \Vert_p$,
 hence $\vert 2n\vert ^{1/p-1/q} \Vert x'-y'\Vert_q \geq \Vert  x'-y' \Vert_p$  
 or $\vert 2n\vert ^{1/p-1/q} \text{DTW}_q(x,y) \geq \Vert  x'-y' \Vert_p$.
Since $x',y'$ represent a valid warping path of $x,y$,
then  $\Vert x'-y' \Vert_p \geq \text{DTW}_p(x,y)$ which concludes the proof. 
\end{proof}

\section{The Triangle Inequality}
\label{sec:ti}
The DTW is commonly used as a similarity measure:  $x$ and $y$ are similar
if $\text{DTW}_p(x,y)$ is small. Similarity measures
often define equivalence relations:  $A\sim A$ for all $A$ (reflexivity), $A\sim B \Rightarrow B \sim C$ (symmetry) and $A\sim B \land B \sim C \Rightarrow A \sim C$ (transitivity).

The DTW is  reflexive and symmetric, but it is not transitive. 
  Indeed, consider the following time series:
  \begin{align*}
  X&= \underbrace{0,0,\ldots,0,0}_{2m+1\text{ times}},\\
  Y& =\underbrace{0,0,\ldots,0,0}_{m\text{ times}}, \epsilon, \underbrace{0,0,\ldots,0,0}_{m\text{ times}},\\
  Z& =0, \underbrace{\epsilon,\epsilon,\ldots,\epsilon,\epsilon}_{2m-1\text{ times}}, 0.  
  \end{align*}
  We have that $\text{NDTW}_p(X,Y)=\text{DTW}_p(X,Y)=\vert \epsilon\vert $, $\text{NDTW}_p(Y,Z)=\text{DTW}_p(Y,Z)=0$,  $\text{NDTW}_p(X,Z)=\text{DTW}_p(X,Z)=\sqrt[p]{(2m-1) }\vert \epsilon\vert$ for $1\leq p<\infty$ and $w=m-1$.
  Hence, for $\epsilon$ small and $n\gg 1/\epsilon$, we have that  $X\sim Y$ and $Y\sim Z$, but $X \not\sim Z$. This example proves
  the following lemma.
 
 \begin{lemma}\label{lemma:ti}For $1\leq p < \infty$ and $w>0$, neither $\text{DTW}_p$ nor $\text{NDTW}_p$  satisfies
 a triangle inequality of the form $d(x,y)+d(y,z)\geq c d(x,z)$ where $c$ is independent of the length of the time series
 and of the locality constraint. 
 \end{lemma}
 
 This theoretical result is somewhat at odd with practical experience.
 Casacuberta et al. found no triangle inequality violation in about 15~million
  triplets of voice recordings~\cite{casacuberta1987mpd}. 
To determine whether we could expect violations of the triangle inequality in
practice, we ran the following experiment. We used 3~types of 100-sample time series: 
 white-noise times series 
defined by $x_i = N(0,1)$ where $N$ is the normal distribution,  random-walk 
time series defined by $x_i = x_{i-1}+N(0,1)$ and $x_1=0$, and the 
Cylinder-Bell-Funnel
time series proposed by Saito~\cite{921732}. For each type, we generated
100~000~triples of time series $x,y,z$ and we computed the histogram of the
function \begin{align*}C(x,y,z)=\frac{\text{DTW}_p(x,z)}{\text{DTW}_p(x,y) + \text{DTW}_p(y,z)} \end{align*} for $p=1$ and $p=2$.
The DTW is computed without time constraints. Over the white-noise and Cylinder-Bell-Funnel time
series, we failed
to find a single violation of the triangle inequality: a triple $x,y,z$ for
which $C(x,y,z) > 1$. 
However, for the random-walk time series,  we found that 20\% and 15\% of the triples violated the triangle inequality
for $\text{DTW}_1$ and $\text{DTW}_2$.



 
 

The DTW satisfies a weak triangle inequality as the next theorem shows.

\begin{theorem}
Given any 3~same-length time series $x,y,z$ and $1\leq p\leq \infty$, we have 
\begin{align*}\text{DTW}_p(x,y)+\text{DTW}_p(y,z) \geq \frac{\text{DTW}_p(x,z)}{\min(2w+1,n)^{1/p}}\end{align*}
where $w$ is the locality constraint. The result also holds for the non-monotonic DTW.
\end{theorem}
\begin{proof} Let $\Gamma$ and $\Gamma'$ be  minimal warping paths between $x$ and $y$
and between $y$ and $z$. Let $\Gamma''= \{(i,j,k) | (i,j) \in \Gamma \text{ and } (j,k) \in \Gamma'\}$.
Iterate through the tuples $ (i,j,k)$  in $\Gamma''$ and construct the same-length
time series $x'', y'', z''$ from $x_i$, $y_j$, and $z_k$. By the locality constraint any match $(i,j)\in \Gamma$
corresponds to at most $\min(2w+1,n)$~tuples of the form $(i,j,\cdot)\in \Gamma''$, and similarly for
any match $(j,k)\in \Gamma'$. 
Assume $1\leq p< \infty$. We
 have that $\Vert x''-y'' \Vert_p^p = \sum_{(i,j,k)\in \Gamma''} \vert x_i-y_j \vert^p \leq \min(2w+1,n)\text{DTW}_p(x,y)^p$
and $\Vert y''-z'' \Vert_p^p =\sum_{(i,j,k)\in \Gamma''} \vert y_j-z_k \vert^p \leq \min(2w+1,n)\text{DTW}_p(y,z)^p$.
By the triangle inequality in $l_p$, we have
\begin{align*}
 \min(2w+1,n)^{1/p}(\text{DTW}_p(x,y)+\text{DTW}_p(y,z)) & \geq \Vert x''-y'' \Vert_p+\Vert y''-z'' \Vert_p\\ 
                                                   &\geq  \Vert x''-z'' \Vert_p\geq \text{DTW}_p(x,z).
 \end{align*}
 For $p=\infty$, $\max_{(i,j,k)\in \Gamma''} \Vert x_i-y_j \Vert_p^p =\text{DTW}_{\infty}(x,y)^p$
 and
 $\max_{(i,j,k)\in \Gamma''} \vert y_j-z_k \vert^p = \text{DTW}_{\infty}(y,z)^p$, thus proving the result
 by the triangle inequality over $l_{\infty}$. 
The proof is the same for the non-monotonic DTW\@.\end{proof}

The constant $\min(2w+1,n)^{1/p}$ is tight. Consider the example with time series $X,Y,Z$ presented before
 Lemma~\ref{lemma:ti}. We have $\text{DTW}_p(X,Y) + \text{DTW}_p(Y,Z)=\vert \epsilon\vert $ 
 and $\text{DTW}_p(X,Z)= \sqrt[p]{(2w+1) }\vert \epsilon\vert $. Therefore, we have
 \begin{align*}\text{DTW}_p(X,Y)+\text{DTW}_p(Y,Z)= \frac{\text{DTW}_p(X,Z)}{\min(2w+1,n)^{1/p}}.\end{align*}

A consequence of this theorem is that $\text{DTW}_\infty$ satisfies the
traditional triangle inequality.
 
\begin{corollary}\label{coro:tilinfty}
The triangle inequality  $d(x,y)+ d(y,z)\geq d(x,z)$ holds
for $\text{DTW}_\infty$ and $\text{NDTW}_\infty$.
\end{corollary}

Hence the $\text{DTW}_\infty$ is a pseudometric: it is a metric over equivalence
classes defined by  $x\sim y$ if and only if $\text{DTW}_\infty(x,y)=0$.
When no locality constraint is enforced, $\text{DTW}_\infty$ is equivalent
to the discrete
Fréchet distance~\cite{eitermannila94}.

\section{Which is the Best Distance Measure?}
\label{sec:whichmeasure}
The DTW can be seen as the minimization of  the $l_p$ distance
under warping. Which $p$ should we choose?
Legrand et al. reported  best results
for chromosome classification using $\text{DTW}_1$~\cite{legrand2007ccu}
as opposed to using $\text{DTW}_2$. However, they did not quantify the
benefits of $\text{DTW}_1$.
Morse and Patel reported  similar results with both 
$\text{DTW}_1$ and $\text{DTW}_2$~\cite{morse2007eaa}.

While they do not consider the DTW, Aggarwal et al.~\cite{656414} 
argue
that out of the usual $l_p$ norms, only the $l_1$ norm, and to a lesser
extend the $l_2$ norm, express a qualitatively meaningful
distance when there are numerous dimensions. They even report on classification-accuracy experiments
where fractional $l_p$ distances such as $l_{0.1}$ and $l_{0.5}$ fare
better.  Fran\c{c}ois et al.~\cite{francois2007cfd} made the theoretical
result more precise showing that under uniformity assumptions, lesser
values of $p$ are always better. 

%







To compare $\text{DTW}_1$, $\text{DTW}_2$, $\text{DTW}_4$ and 
$\text{DTW}_{\infty}$, we considered four different synthetic time-series data sets:
Cylinder-Bell-Funnel~\cite{921732},
  Control Charts~\cite{pham1998ccp},
 Waveform~\cite{breiman1998car}, and
Wave+Noise~\cite{gonzalez2000tsc}. The time series in each data sets have
lengths 128, 60, 21, and 40.
 The Control Charts data set has
6~classes of time series whereas the other 3~data sets have 3~classes each.
For each data set, we generated various databases having
a different number of instances per class: between 1 and 9
inclusively for Cylinder-Bell-Funnel and Control Charts, and
between 1 and 99 for Waveform and Wave+Noise. For a given
data set and a given number of instances, 50~different databases
were generated. For each
database, we generated 500~new instances chosen from a random 
class and we found a nearest neighbor in the database using $\text{DTW}_p$
for $p=1,2,4,\infty$ and using a time constraint of $w=n/10$. When
the instance is of the same class as the nearest neighbor,
we considered that the classification was a success.

The average
classification accuracies for the 4~data sets, and for various
number of instances per class is given in Fig.~\ref{fig:classaccuracy}.
The average is taken over 25~000~classification tests ($50\times 500$), over 50~different
databases.

\begin{figure*}
\centering
  \subfloat[Cylinder-Bell-Funnel]{\includegraphics[width=0.45\textwidth]{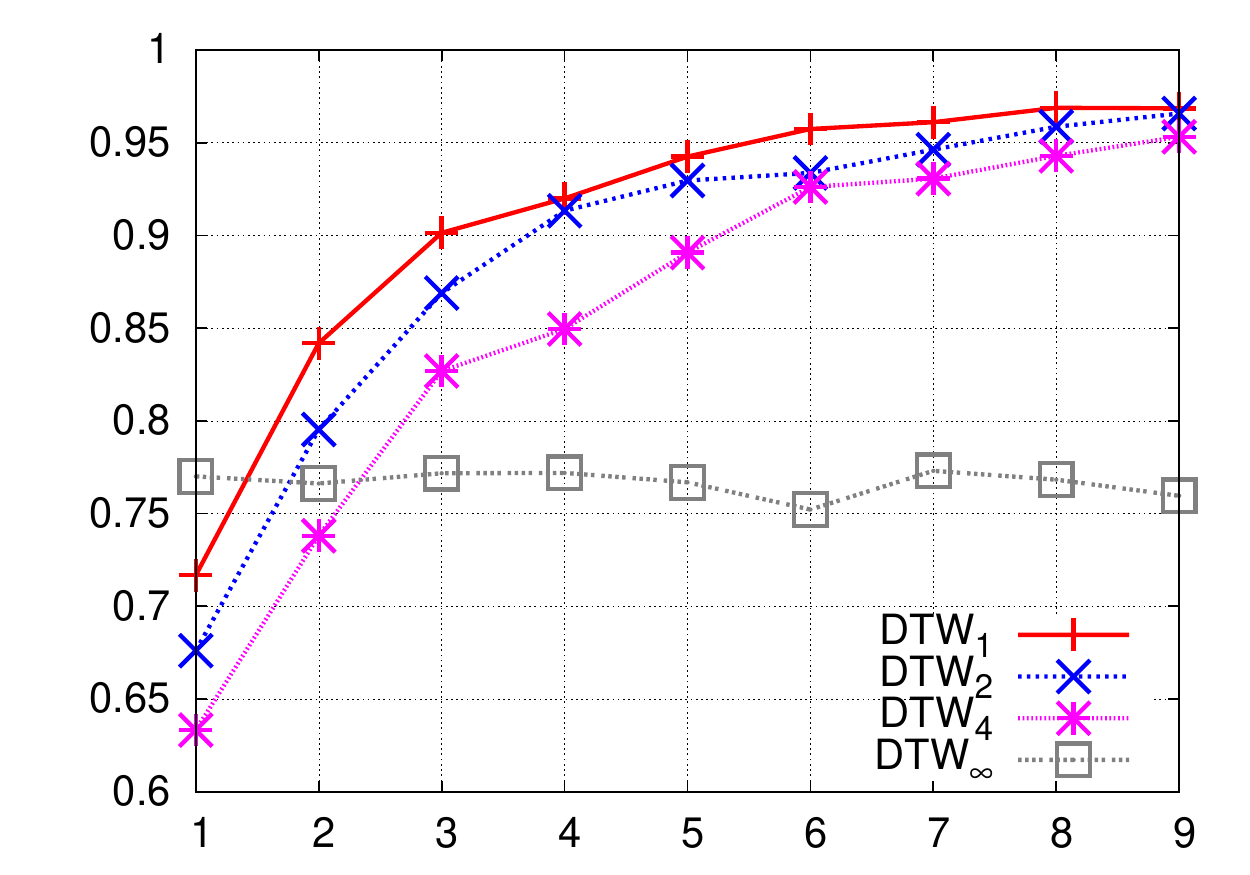}}
  \subfloat[Control Charts]{\includegraphics[width=0.45\textwidth]{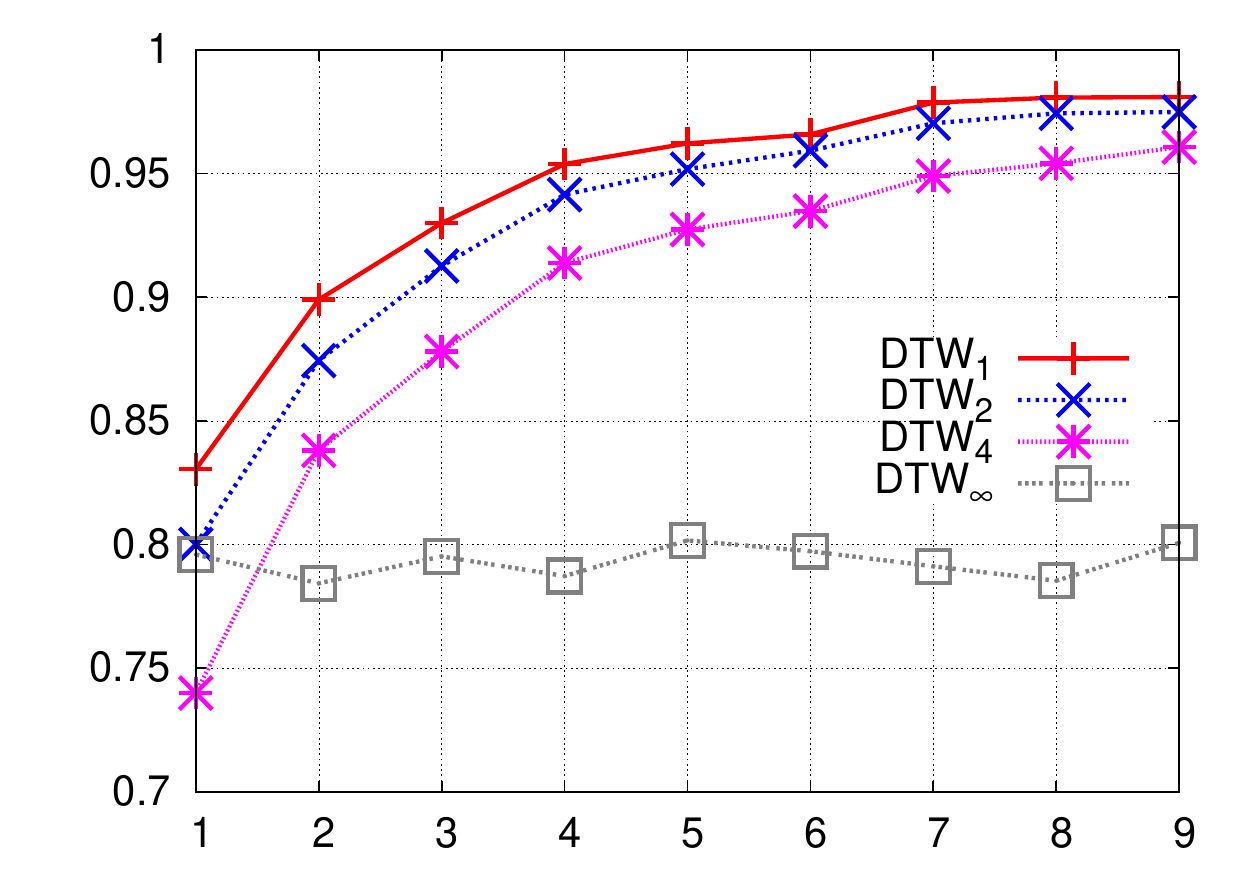}} \\
\subfloat[Waveform]{\includegraphics[width=0.45\textwidth]{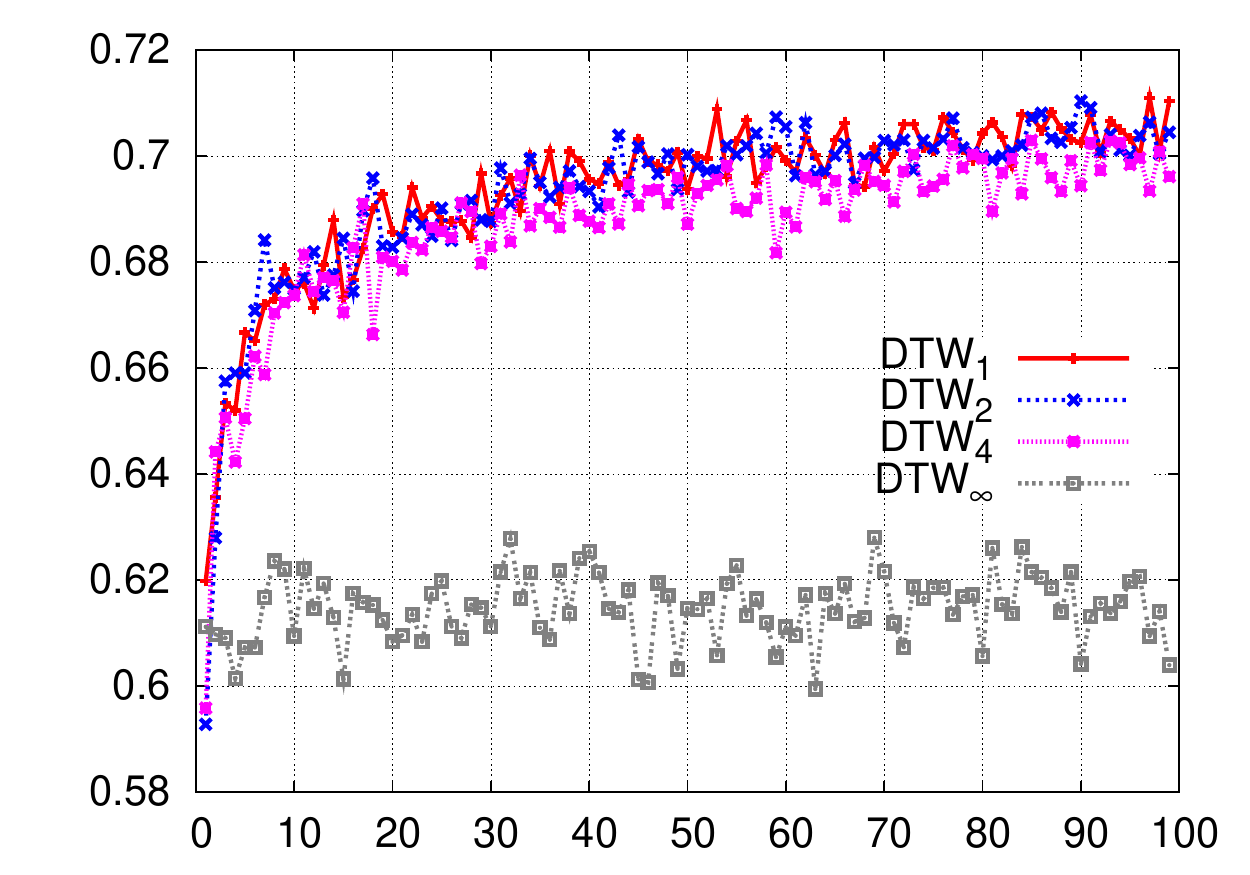}} 
  \subfloat[Wave+Noise]{\includegraphics[width=0.45\textwidth]{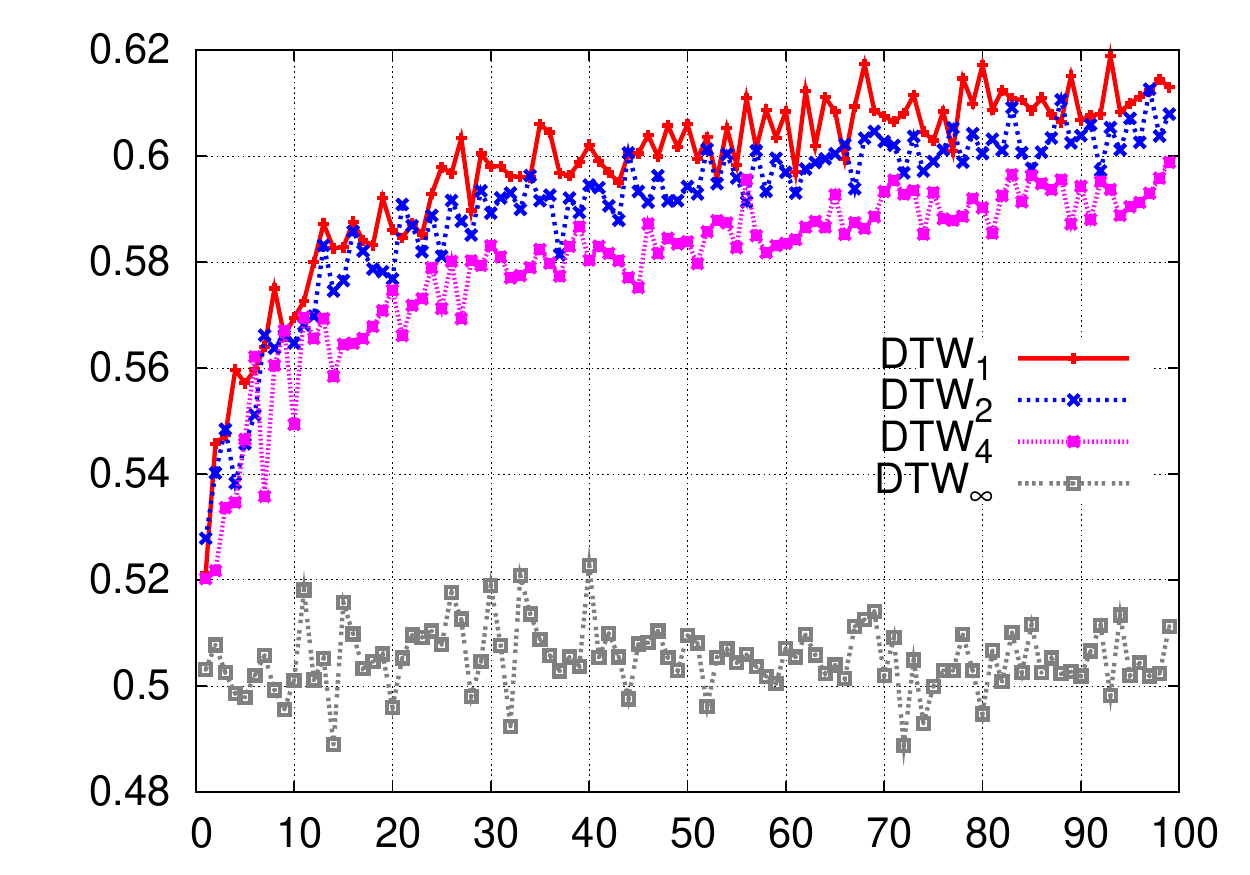}}
\caption{Classification accuracy versus the number of instances of each class in four data sets\label{fig:classaccuracy}}
\end{figure*}

Only when there are one or two instances of each class is $\text{DTW}_{\infty}$
competitive. Otherwise, the accuracy of the $\text{DTW}_{\infty}$-based classification
does not improve as we add more instances of each class. For the 
Waveform data set, $\text{DTW}_1$ and $\text{DTW}_2$ have comparable 
accuracies. 
For the other 3~data sets, 
$\text{DTW}_1$ has  a better nearest-neighbor classification
accuracy than $\text{DTW}_2$. 
Classification with $\text{DTW}_4$ has almost always a lower accuracy than
either $\text{DTW}_1$  or 
$\text{DTW}_2$.

Based on these results, $\text{DTW}_1$ is a good choice  to classify time series
whereas $\text{DTW}_2$ is a close second.  

\section{Computing Lower Bounds on the DTW}

\label{sec:definingUL}
Given a time series $x$, define $U(x)_i = \max_k \{x_k | \,\vert k-i \vert \leq w \}$
and $L(x)_i = \min_k \{x_k |\, \vert k-i \vert \leq w \}$ for $i=1, \ldots, n$. The pair $U(x)$ and $L(x)$ forms
the warping envelope of $x$ (see Fig.~\ref{fig:pykeogh-4}). We leave the time constraint $w$ implicit.

\begin{figure}
\centering
\includegraphics[width=0.7\columnwidth]{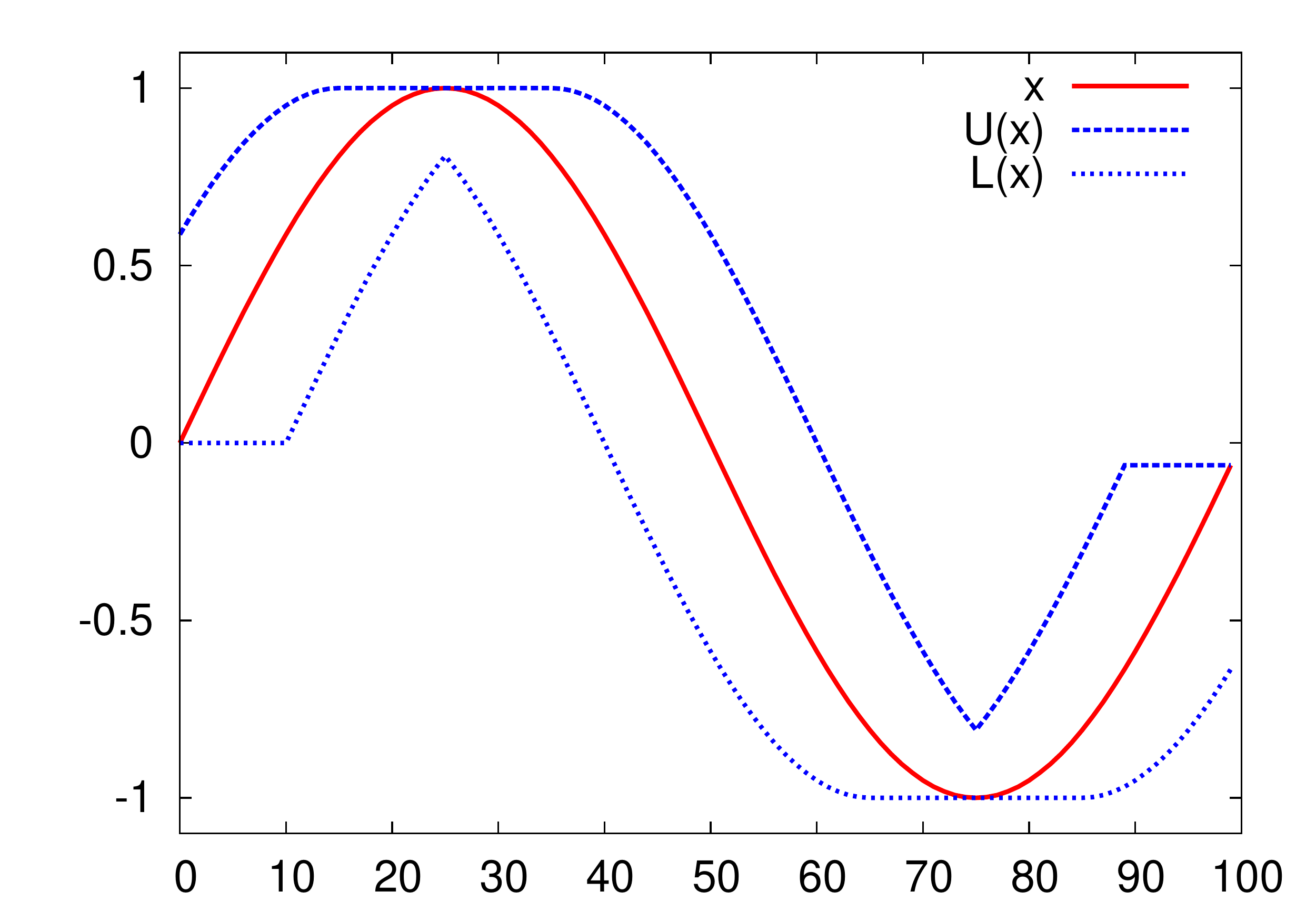}
\caption{\label{fig:pykeogh-4}Warping envelope example}
\end{figure}

The theorem of this section has an elementary proof  requiring only
the following technical lemma.

\begin{lemma}\label{lemma:technical}If $b\in[a,c]$ then $(c-a)^p \geq (c-b)^p+(b-a)^p$ for $1\leq p< \infty$.
\end{lemma}
\begin{proof}For $p=1$, $(c-b)^p+(b-a)^p=(c-a)^p$. For $p>1$, by deriving
  $(c-b)^p+(b-a)^p$ with respect to $b$, we can show that it is minimized when 
$b=(c+a)/2$ and maximized when $b\in \{a,c\}$. The maximal value is $(c-a)^p$. Hence the result.\end{proof}

The following theorem introduces a generic result that we  use to derive two
lower bounds for the DTW including the original Keogh-Ratanamahatana result~\cite{keogh2005eid}.
%

\begin{theorem}\label{thm:mainthm}
Given two equal-length time series $x$ and $y$ and $1\leq p<\infty$, then for any time series
$h$ satisfying $x_i\geq  h_i\geq U(y)_i$ or  $x_i \leq h_i \leq L(y)_i$
or $h_i=x_i$ for all indexes $i$, we have
\begin{align*}
\text{DTW}_p(x,y)^p & \geq \text{NDTW}_p(x,y)^p \\
                    & \geq   \Vert x-h \Vert_p^p + \text{NDTW}_p(h,y)^p.
\end{align*}
For $p=\infty$, a similar result is true: $\text{DTW}_{\infty}(x,y)  \geq \text{NDTW}_{\infty}(x,y)  \geq   \max(\Vert x-h \Vert_{\infty} ,\text{NDTW}_{\infty}(h,y))$.
\end{theorem}
\begin{proof}
Suppose that $1\leq p<\infty$.
Let $\Gamma$ be a warping path such that $\text{NDTW}_p(x,y)^p = \sum_{(i,j) \in \Gamma} \vert x_i-y_j \vert_p^p $.
By the constraint on $h$ and Lemma~\ref{lemma:technical}, we have that $\vert x_i-y_j \vert^p  \geq \vert x_i-h_i \vert^p  + \vert h_i - y_j \vert^p $
for any $(i,j) \in \Gamma$ since $h_i \in [\min(x_i, y_j), \max(x_i,y_j)]$. 
Hence, we have that  $\text{NDTW}_p(x,y)^p \geq \sum_{(i,j) \in \Gamma}  \vert x_i-h_i \vert^p  + \vert h_i - y_j \vert^p
\geq  \Vert x-h \Vert_p^p  + \sum_{(i,j) \in \Gamma} \vert h_i - y_j \vert^p$.
This proves the result since  $ \sum_{(i,j) \in \Gamma} \vert h_i - y_j \vert \geq  \text{NDTW}_p(h,y)$.
For $p=\infty$, we have that $\text{NDTW}_{\infty}(x,y) =\max_{(i,j)\in \Gamma} \vert x_i-y_j \vert \leq 
\max_{(i,j)\in \Gamma} \max(\vert x_i-h_i \vert,\vert h_i-y_j \vert)=\max(\Vert x-h \Vert_{\infty} ,\text{NDTW}_{\infty}(h,y))$,
concluding the proof.
\end{proof}

While Theorem~\ref{thm:mainthm} defines a lower bound ($\Vert x-h \Vert_p$), the next proposition shows that this
lower bound must be a tight approximation as long as $h$ is close to $y$ in the $l_p$ norm.
\begin{proposition}\label{prop:error}
Given two equal-length time series $x$ and $y$, and $1\leq p\leq \infty$ with $h$ as in Theorem~\ref{thm:mainthm},
we have that $\Vert x-h \Vert_p$ approximates both $ \text{DTW}_p(x,y)$
and $ \text{NDTW}_p(x,y)$
within $\Vert h-y \Vert_p$.
\end{proposition}
\begin{proof}
By the triangle inequality over $l_p$, we have  $\Vert x-h \Vert_p +   \Vert h-y \Vert_p \geq\Vert x-y \Vert_p $. 
Since $\Vert x-y \Vert_p \geq  \text{DTW}_p(x,y)$,
we have
$\Vert x-h \Vert_p +   \Vert h-y \Vert_p \geq \text{DTW}_p(x,y)$,
and hence
$ \Vert h-y \Vert_p \geq \text{DTW}_p(x,y) - \Vert x-h \Vert_p $.
This proves the result since by
 Theorem~\ref{thm:mainthm},
we have that $\text{DTW}_p(x,y)\geq \text{NDTW}_p(x,y) \geq \Vert x-h \Vert_p$.
\end{proof}

This bound on the approximation error is reasonably tight. If $x$ and $y$ 
are separated by a constant, then $\text{DTW}_1(x,y) =\Vert x-y \Vert_1$
by Proposition~\ref{prof:band1} and $\Vert x-y \Vert_1 = \sum_i \vert x_i-y_i\vert
=  \sum_i \vert x_i-h_i\vert+\vert h_i-y_i\vert=\Vert x-h \Vert_1+ \Vert h-y \Vert_1 $.
Hence, the approximation error is exactly $\Vert h-y \Vert_1 $ in such instances.

%
%

\section{Warping Envelopes}
\label{sec:envelopes}
The computation of the warping envelope $U(x),L(x)$ requires $O(nw)$ time using the naive approach
of repeatedly computing the maximum and the minimum over windows. 
Instead, we compute the envelope using at most $3n$~comparisons between data-point values~\cite{lemiremaxmin}
using Algorithm~\ref{algo:mystreamingalo}.

\begin{algorithm}
\begin{small}
 \begin{algorithmic}
\STATE \textbf{input} a time series $a$ indexed from $1$ to $n$
\STATE \textbf{input} some DTW time constraint $w$
\RETURN warping envelope $U,L$ (two time series of length $n$)
\STATE $u$, $l$ $\leftarrow$ empty double-ended queues, we append to ``back'' 
\STATE append $1$ to $u$ and $l$ 
\FOR{ $i$ in $\{2,\ldots,n\}$}\label{alg:mainloop}
\IF{$i\geq w+1$}
\STATE $U_{i-w} \leftarrow a_{\textrm{front}(u)}$, $L_{i-w} \leftarrow a_{\textrm{front}(l)}$ 
\ENDIF
\IF{$a_i > a_{i-1}$}\label{alg:firstcompare}
\STATE pop $u$ from back\label{alg:removemax}
\WHILE{ $a_i > a_{\textrm{back}(u)}$}\label{alg:while1}
\STATE pop $u$ from back 
\ENDWHILE
\ELSE
\STATE pop $l$ from back \label{alg:removemin}
\WHILE{ $a_i < a_{\textrm{back}(l)}$} \label{alg:while2}
\STATE pop $l$ from back
\ENDWHILE
\ENDIF
\STATE append $i$ to $u$ and $l$\label{alg:append}
\IF{$i=2w+1+\textrm{front}(u)$}  \label{alg:ensurescontainted}
\STATE pop $u$ from front
\ELSIF{$i=2w+1+\textrm{front}(l)$}
\STATE pop $l$ from front
\ENDIF
\ENDFOR
\FOR{ $i$ in $\{n+1,\ldots,n+w\}$}\label{alg:secondloop}
\STATE $U_{i-w} \leftarrow a_{\textrm{front}(u)}$, $L_{i-w} \leftarrow a_{\textrm{front}(l)}$
\IF{i-front($u$)$\geq 2w+1$}
\STATE pop $u$ from front
\ENDIF
\IF{i-front($l$)$\geq 2w+1$}
\STATE pop $l$ from front
\ENDIF
\ENDFOR
 \end{algorithmic}
\end{small}
\caption{\label{algo:mystreamingalo}Streaming algorithm to compute the warping envelope using no more than $3n$~comparisons
}
\end{algorithm}

Envelopes are asymmetric in the sense that if $x$ is in the envelope of $y$ ($L(y) \leq x \leq U(x)$),
it does not follow that $x$ is in the envelope of $y$ ($L(x) \leq y \leq U(x)$). For example, $x=0,0,\ldots,0$
is in the envelope of $y=1,-1,1,-1,1,-1,\ldots,1$ for $w>1$, but the reverse is not true.
However, the next lemma shows that if $x$ is below or above the envelope of $y$, then
$y$ is above or below the envelope of $x$.


\begin{lemma}\label{lemma:mutualenv}$L(x) \geq y$ is equivalent to $x \geq U(y)$.\label{lemma:ul1}
\end{lemma}
\begin{proof}
Suppose $x_i < U(y)_i$ for some $i$, then there is
$j$ such that $|i-j| \leq w$ and $x_i < y_j$, therefore $L(x)_j < y_j$.
It follows that $L(x) \geq y$ implies $x \geq U(y)$. The reverse
implication follows similarly.
\end{proof}



We know that $L(h)$ is less or equal to  $h$ whereas $U(h)$ is greater or equal to $h$. The next
lemma shows that $U(L(h))$ is less or equal than $h$ whereas $L(U(h))$ is greater or equal than $h$.

\begin{lemma}We have $U(L(h))\leq h$ and  $L(U(h)) \geq h$ for any  $h$.\label{lemma:ul2}
\end{lemma}
\begin{proof}By definition, we have that $L(h)_j\leq h_i$ whenever $\vert i-j\vert \leq w$. Hence,
$\max_{j| \vert i-j\vert \leq w} L(h)_j\leq h_i$ which proves $U(L(h))\leq h$. The second
result ($L(U(h)) \geq h$) follows similarly.
\end{proof}

Whereas $L(U(h))$ is greater or equal than $h$, the next lemma shows that $U(L(U(h)))$ is
equal to $U(h)$.

\begin{corollary}We have $U(h)=U(L(U(h)))$  and $L(h) = L(U(L(h)))$ for any $h$.\end{corollary}
\begin{proof}By Lemma~\ref{lemma:ul2}, we have $L(U(h)) \geq h$, hence $U(L(U(h))) \geq U(h)$.
Again by  Lemma~\ref{lemma:ul2}, we have $U(L(h'))\leq h'$ for $h'=U(h)$ or  $U(L(U(h)))\leq U(h)$.
Hence, $U(h)=U(L(U(h)))$. The next result ($L(h) = L(U(L(h)))$) follows similarly.
\end{proof}

\section{LB\_Keogh}
\label{sec:lbkeogh}
Let $H(x,y)$ be the \emph{projection of $x$ on $y$} defined as
\begin{align}\label{eqn:hxy}
H(x,y)_{i}= \begin{cases}U(y)_i  & \text{if $x_i\geq U(y)_i$}\\
L(y)_i  & \text{if $x_i\leq L(y)_i$} \\
x_i & \text{otherwise,}
\end{cases}
\end{align}
for $i=1,2,\ldots,n$.
We have that $H(x,y)$ is in the envelope of $y$. By Theorem~\ref{thm:mainthm} and setting $h= H(x,y)$, 
we have that $\text{NDTW}_p(x,y)^p\geq \Vert x- H(x,y)\Vert_p^p + NDTW(H(x,y),y)^p$ for $1 \leq p<\infty$.
Write $\text{LB\_Keogh}_p(x,y)=\Vert x-H(x,y) \Vert_p$ (see Fig.~\ref{fig:pykeogh-2}).
The following corollary follows  from Theorem~\ref{thm:mainthm} and Proposition~\ref{prop:error}.

\begin{figure}
\centering
\includegraphics[width=0.7\columnwidth]{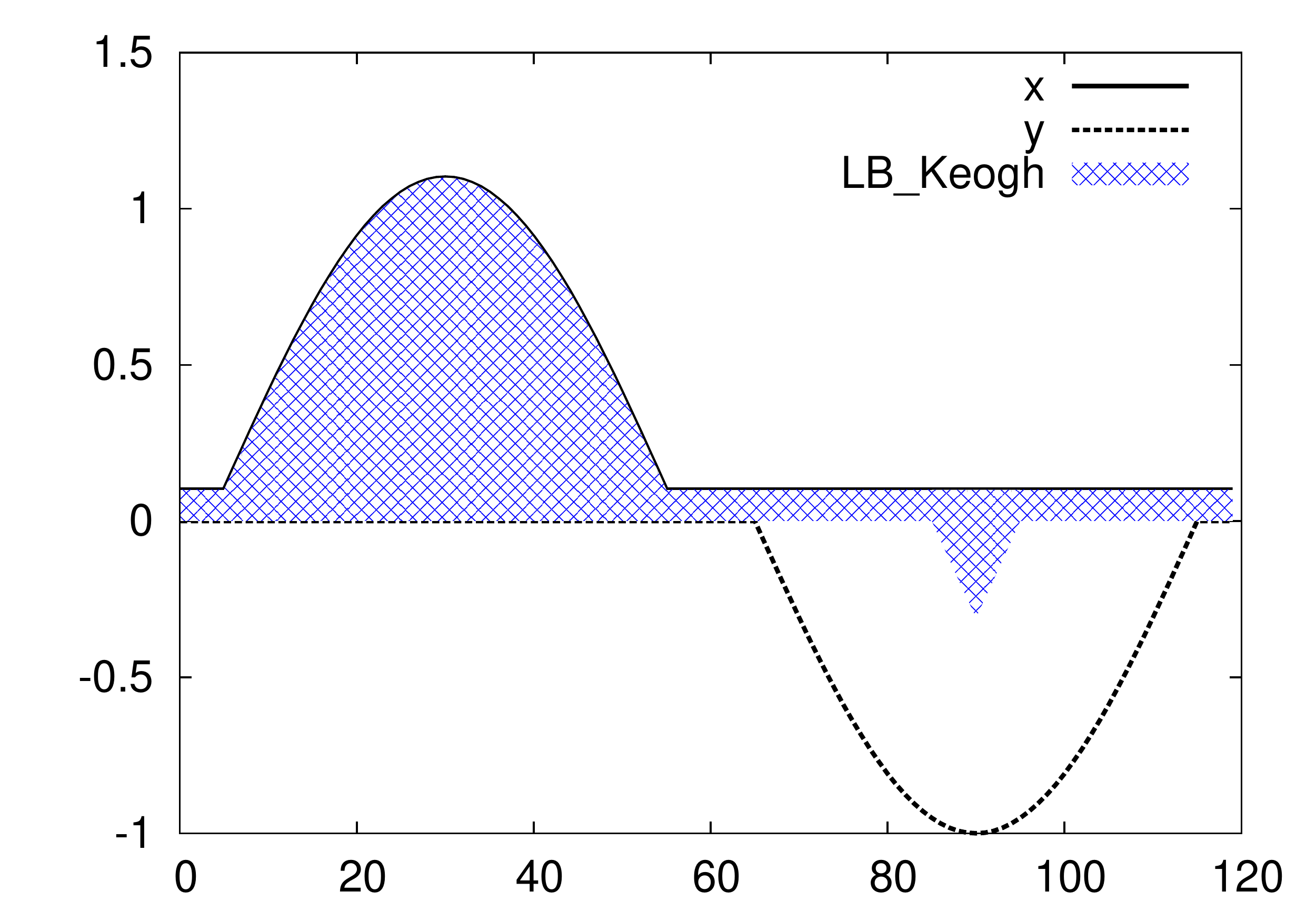}
\caption{\label{fig:pykeogh-2}LB\_Keogh example: the area of the marked region is $\text{LB\_Keogh}_1(x,y)$}
\end{figure}

\begin{corollary}
Given two equal-length time series $x$ and $y$ and $1\leq p\leq\infty$, then
\begin{itemize}
\item $\text{LB\_Keogh}_p(x,y)$  is a lower bound
to the DTW: \begin{align*}\text{DTW}_p(x,y)  \geq \text{NDTW}_p(x,y) \geq   \text{LB\_Keogh}_p(x,y);\end{align*}
\item the accuracy of LB\_Keogh is bounded by the distance to the envelope: \begin{align*}\text{DTW}_p(x,y)-\text{LB\_Keogh}_p(x,y) \leq \Vert \max \{U(y)_i - y_i, y_i-L(y)_i\}_i \Vert_p\end{align*} for all $x$. 
\end{itemize}
\end{corollary}

%

Algorithm~\ref{algo:lbkeogh} shows how LB\_Keogh can be used
to find a nearest neighbor in a time series database.
The computation of the envelope of the query time series is done
once (see line~\ref{line:mylbkeoghenvelope}).
The lower bound is computed in lines~\ref{line:mylbkeoghstart}
to \ref{line:mylbkeoghend}. If the lower bound is sufficiently
large, the DTW is not computed (see line~\ref{line:mylbkeoghtest}). 
Ignoring the computation of the full DTW, at most
$(2N+3)n$~comparisons between data points are required to process a database
containing $N$~time series.

\begin{algorithm}
\begin{small}
 \begin{algorithmic}[1]
\STATE \textbf{input} a time series $a$ indexed from $1$ to $n$
\STATE \textbf{input} a set $S$ of candidate time series
\RETURN the nearest neighbor $B$ to $a$ in $S$ under $\text{DTW}_1$
\STATE $U,L \leftarrow \text{envelope}(a)$\label{line:mylbkeoghenvelope}
\STATE $b\leftarrow \infty$
\FOR {candidate $c$ in $S$}
\STATE $\beta \leftarrow 0$\label{line:mylbkeoghstart}
\FOR{$i \in \{1,2,\ldots,n\}$}
\IF{$c_i > U_i$}
\STATE $\beta \leftarrow \beta + c_i-U_i$
\ELSIF{$c_i < L_i$}
\STATE $\beta \leftarrow \beta + L_i-c_i$\label{line:mylbkeoghend}
\ENDIF
\ENDFOR
\IF{$\beta < b$}\label{line:mylbkeoghtest}
\STATE $t\leftarrow \text{DTW}_1(a,c)$
\IF{$t < b$}
\STATE $b\leftarrow t$
\STATE $B\leftarrow c$
\ENDIF
\ENDIF
\ENDFOR
 \end{algorithmic}
\end{small}
\caption{\label{algo:lbkeogh}LB\_Keogh-based Nearest-Neighbor algorithm
}
\end{algorithm}

\section{LB\_Improved}

\label{sec:lbimproved}

Write $\text{LB\_Improved}_p(x,y)^p=\text{LB\_Keogh}_p(x,y)^p+\text{LB\_Keogh}_p(y,H(x,y))^p$
for $1 \leq p<\infty$. 
By definition, we have  $\text{LB\_Improved}_p(x,y)\geq \text{LB\_Keogh}_p(x,y)$.
Intuitively, whereas $\text{LB\_Keogh}_p(x,y)$ measures the distance between $x$ and the envelope of $y$,
$\text{LB\_Keogh}_p(y,H(x,y))$ measures the distance between $y$ and the envelope of the projection
of $x$ on $y$  (see Fig.~\ref{fig:pykeogh-3}).
The next corollary shows that LB\_Improved is a lower bound to the DTW\@.

\begin{figure}
\centering
\includegraphics[width=0.7\columnwidth]{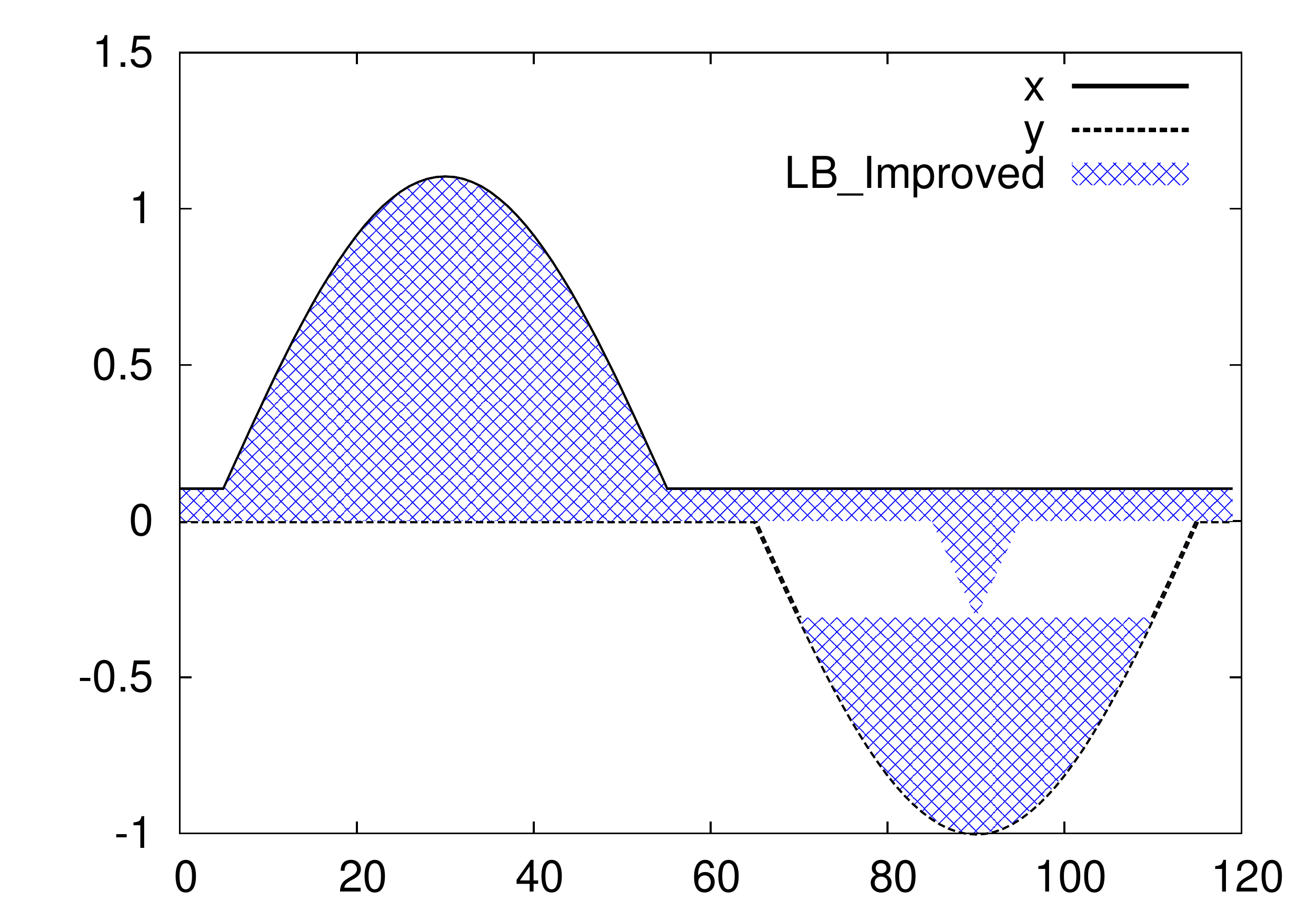}
\caption{\label{fig:pykeogh-3}LB\_Improved example: the area of the marked region is $\text{LB\_Improved}_1(x,y)$}
\end{figure}

\begin{corollary}
Given two equal-length time series $x$ and $y$ and $1\leq p<\infty$, then
 $\text{LB\_Improved}_p(x,y)$  is a lower bound
to the DTW: $\text{DTW}_p(x,y)  \geq \text{NDTW}_p(x,y) \geq   \text{LB\_Improved}_p(x,y)$.
\end{corollary}
\begin{proof}
%
Recall that $\text{LB\_Keogh}_p(x,y)=\Vert x-H(x,y) \Vert_p$. 
First apply Theorem~\ref{thm:mainthm}: $ \text{DTW}_p(x,y)^p  \geq \text{NDTW}_p(x,y)^p  \geq   \text{LB\_Keogh}_p(x,y)^p + \text{NDTW}_p(H(x,y),y)^p$.
Apply Theorem~\ref{thm:mainthm} once more:
$ \text{NDTW}_p(y,H(x,y))^p  \geq   \text{LB\_Keogh}_p(y,H(x,y))^p$. 
By substitution, we get
$ \text{DTW}_p(x,y)^p  \geq \text{NDTW}_p(x,y)^p  \geq   \text{LB\_Keogh}_p(x,y)^p + \text{LB\_Keogh}_p(y,H(x,y))^p $ thus proving the result.
\end{proof}


Algo.~\ref{algo:lbimproved} shows how to apply LB\_Improved as a two-step process. Initially, for each candidate $c$,
we compute the lower bound $\text{LB\_Keogh}_{1}(c,a)$ (see lines~\ref{line:lbkeoghstart} to \ref{line:lbkeoghend}). If this lower bound  is sufficiently large, the candidate is 
discarded (see line~\ref{line:lbkeoghdiscard}), otherwise we add $\text{LB\_Keogh}_1(a,H(c,a))$ to $\text{LB\_Keogh}_{1}(c,a)$, in effect computing $\text{LB\_Improved}_1(c,a)$ (see lines~\ref{line:lbimprovedstart} to \ref{line:lbimprovedend}). If this larger lower bound is sufficiently large, the candidate
is finally discarded (see line~\ref{line:lbimproveddiscard}). Otherwise, we compute the full DTW\@.
If $\alpha$ is the fraction of candidates pruned by LB\_Keogh, at most
$(2N+3)n+5(1-\alpha)Nn$~comparisons between data points are required to process a database
containing $N$~time series.

\begin{algorithm}
\begin{small}
 \begin{algorithmic}[1]
\STATE \textbf{input} a time series $a$ indexed from $1$ to $n$
\STATE \textbf{input} a set $S$ of candidate time series
\RETURN the nearest neighbor $B$ to $a$ in $S$ under $\text{DTW}_1$
\STATE $U,L \leftarrow \text{envelope}(a)$
\STATE $b\leftarrow \infty$
\FOR {candidate $c$ in $S$}
\STATE copy $c$ to $c'$\label{line:copy}
\STATE $\beta \leftarrow 0$\label{line:lbkeoghstart}
\FOR{$i \in \{1,2,\ldots,n\}$}
\IF{$c_i > U_i$}
\STATE $\beta \leftarrow \beta + c_i-U_i$
\STATE $c'_i = U_i$\label{line:modif1}
\ELSIF{$c_i < L_i$}
\STATE $\beta \leftarrow \beta + L_i-c_i$
\STATE $c'_i = L_i$\label{line:lbkeoghend}\label{line:modif2}
\ENDIF
\ENDFOR
\IF{$\beta < b$}\label{line:lbkeoghdiscard}
\STATE $U',L' \leftarrow \text{envelope}(c')$\label{line:lbimprovedstart}
\FOR{$i \in \{1,2,\ldots,n\}$}
\IF{$a_i > U'_i$}
\STATE $\beta\leftarrow \beta + a_i- U'_i$
\ELSIF{$a_i < L'_i$}
\STATE $\beta\leftarrow \beta + L'_i- a_i$\label{line:lbimprovedend}
\ENDIF
\ENDFOR
\IF{$\beta < b$}\label{line:lbimproveddiscard}
\STATE $t\leftarrow \text{DTW}_1(a,c)$
\IF{$t < b$}
\STATE $b\leftarrow t$
\STATE $B\leftarrow c$
\ENDIF
\ENDIF
\ENDIF
\ENDFOR
 \end{algorithmic}
\end{small}
\caption{\label{algo:lbimproved}LB\_Improved-based Nearest-Neighbor algorithm
}
\end{algorithm}


 \section{Comparing LB\_Keogh and LB\_Improved}
\label{sec:comparing}
%
%
%
%
 %

In this section, we benchmark Algorithms~\ref{algo:lbkeogh}~and~\ref{algo:lbimproved}.
We know that the LB\_Improved approach  has at least the pruning power of the LB\_Keogh-based approach,
but does more pruning translate into a faster nearest-neighbor retrieval under the  DTW distance?

We implemented the algorithms in C++ and called the functions from Python scripts.
We used the GNU GCC 4.0.2 compiler on  an Apple Mac~Pro, 
having two  Intel Xeon dual-core processors running at 2.66\,GHz with 2\,GiB of RAM\@. 
All data was loaded in memory before the experiments, and no thrashing was observed.
We measured  the wall clock total time.
In all experiments, we benchmark
nearest-neighbor retrieval under the $\text{DTW}_1$ with the locality constraint $w$  set at 10\% ($w=n/10$).
To ensure reproducibility, our source code is freely available~\cite{googlelbimproved},
including the script used to generate synthetic data sets. We compute the
full DTW using a straight-forward $O(n^2)$-time dynamic programming algorithm.

\subsection{Synthetic data sets}
 
We tested our algorithms using the  Cylinder-Bell-Funnel~\cite{921732} and
 Control Charts~\cite{pham1998ccp} data sets, as well as over a database of random walks.
 We generated 1~000-sample random-walk time series using the formula   $x_i = x_{i-1}+N(0,1)$ and $x_1=0$.
 Results for the Waveform and Wave+Noise data sets are  similar and omitted.
 
 For each data set, we generated a database of 10~000~time series by adding 
 randomly chosen items. The order of the candidates is thus random.
 Fig.~\ref{fig:perf1}, \ref{fig:perf2} and \ref{fig:perf3} show
 the average timings and pruning ratio averaged over 20~queries based on 
 randomly chosen
 time series as we consider larger and large fraction of the database.
 LB\_Improved prunes between 2 and 4 times more candidates and it is faster by a factor between 1.5 and 3.

  
\begin{figure*}
\centering
  \subfloat[Average Retrieval Time]{\includegraphics[width=0.45\textwidth]{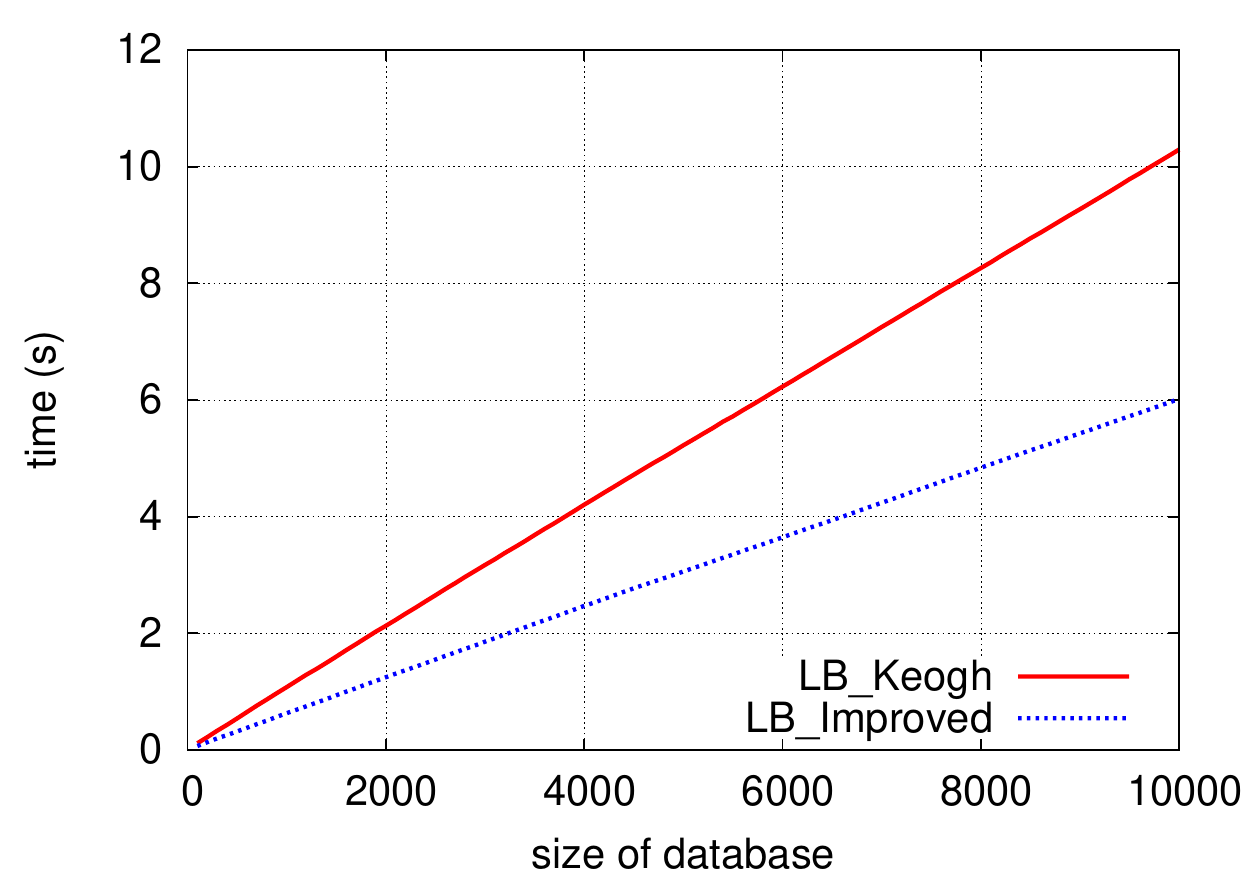}}
  \subfloat[Pruning Power]{\includegraphics[width=0.45\textwidth]{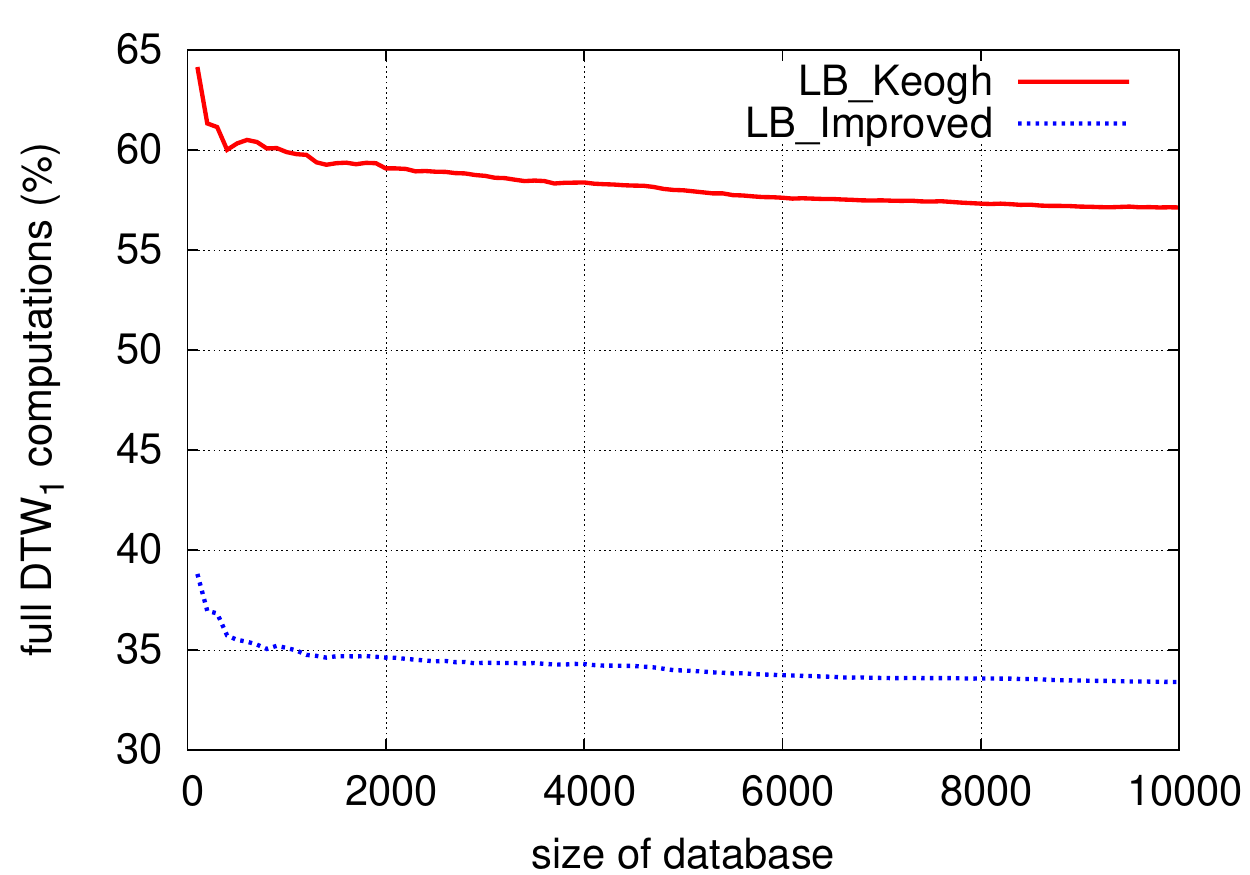}} 
\caption{Nearest-Neighbor Retrieval for the Cylinder-Bell-Funnel data set\label{fig:perf1}}
\end{figure*}

\begin{figure*}
\centering
  \subfloat[Average Retrieval Time]{\includegraphics[width=0.45\textwidth]{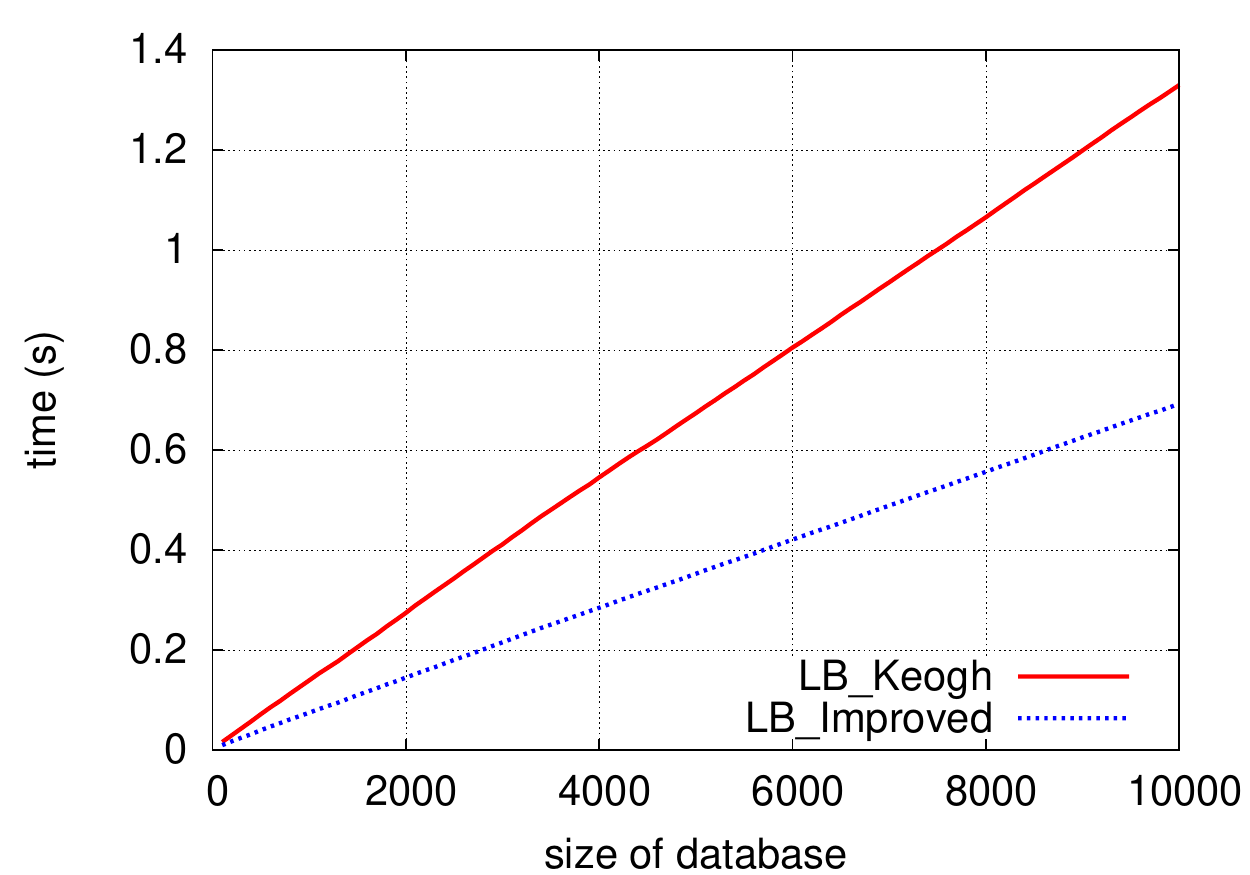}}
  \subfloat[Pruning Power]{\includegraphics[width=0.45\textwidth]{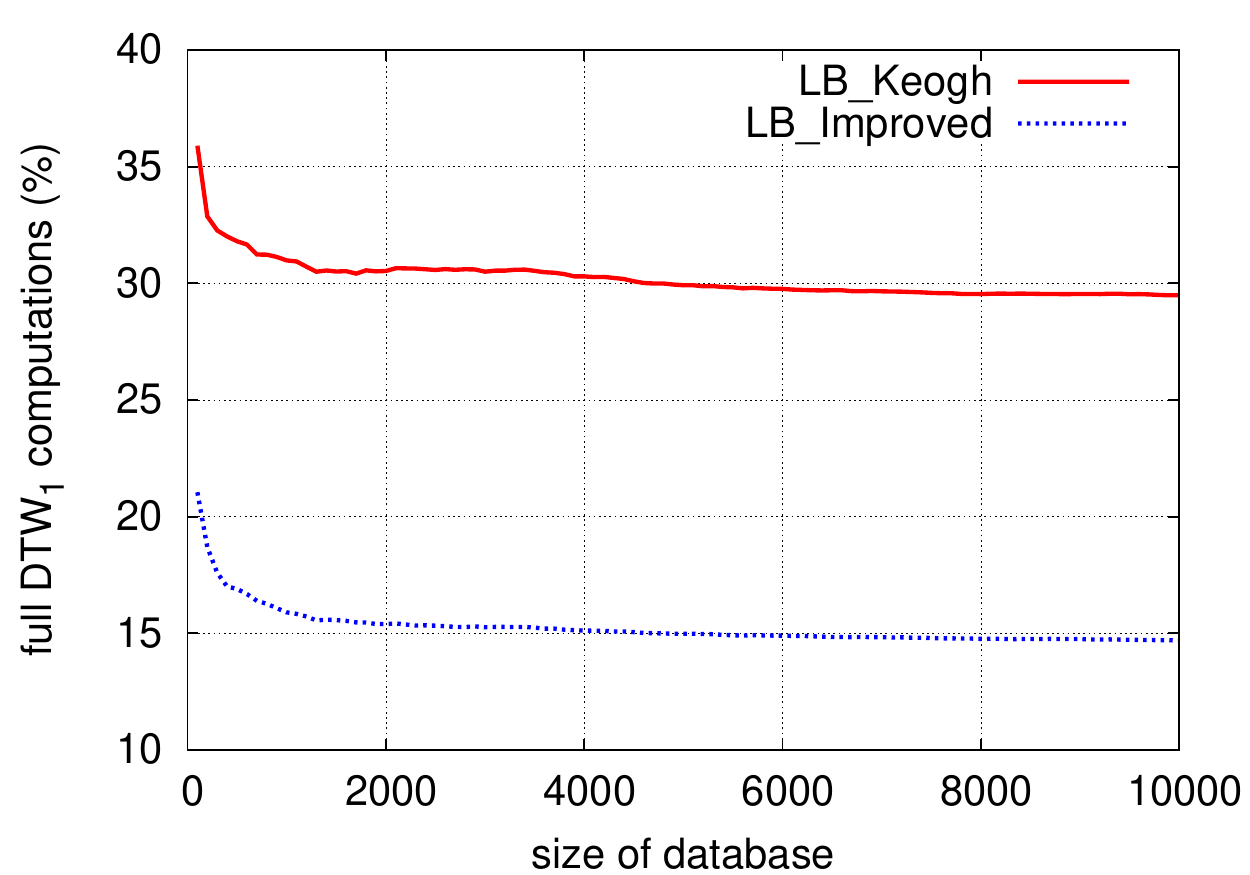}} 
\caption{Nearest-Neighbor Retrieval for the Control Charts data set\label{fig:perf2}}
\end{figure*}

\begin{figure*}
\centering
  \subfloat[Average Retrieval Time]{\includegraphics[width=0.45\textwidth]{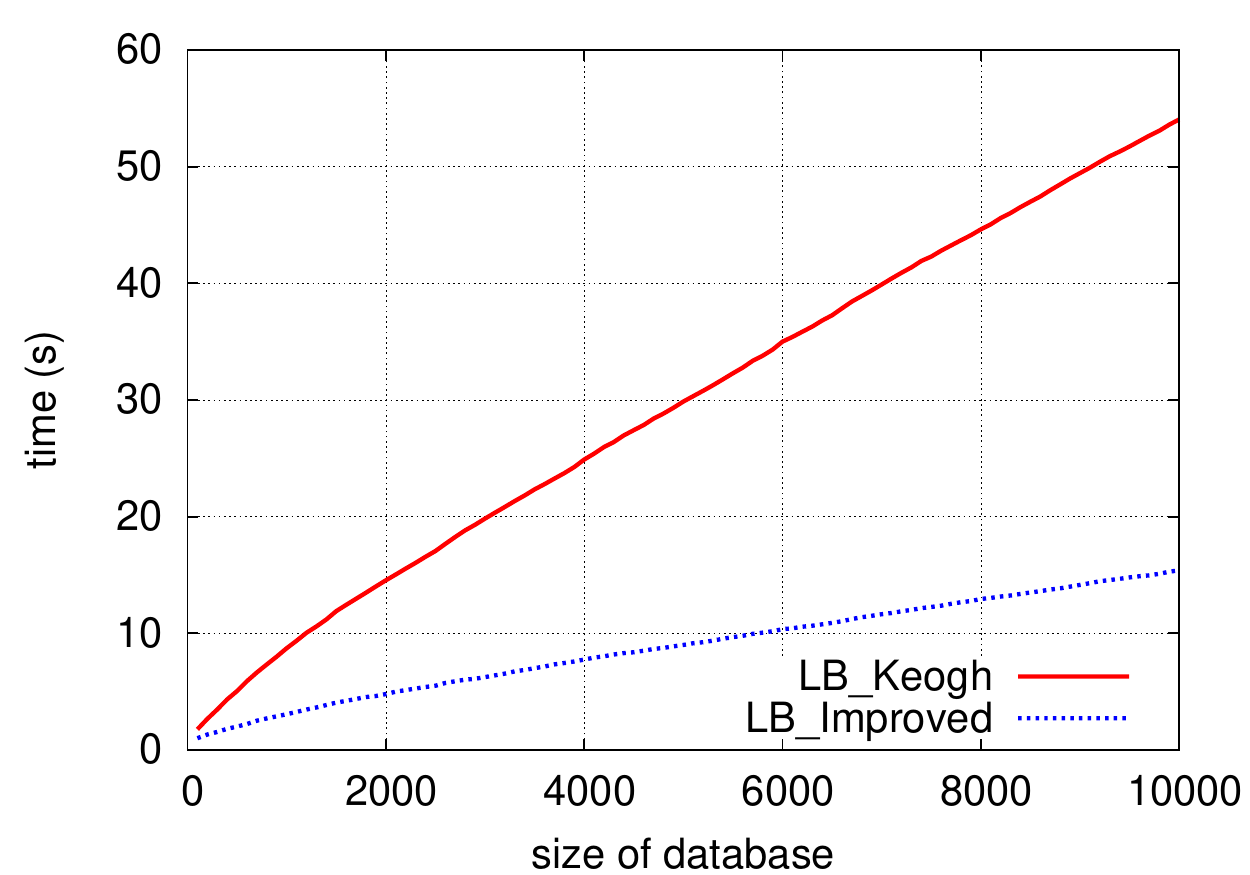}}
  \subfloat[Pruning Power]{\includegraphics[width=0.45\textwidth]{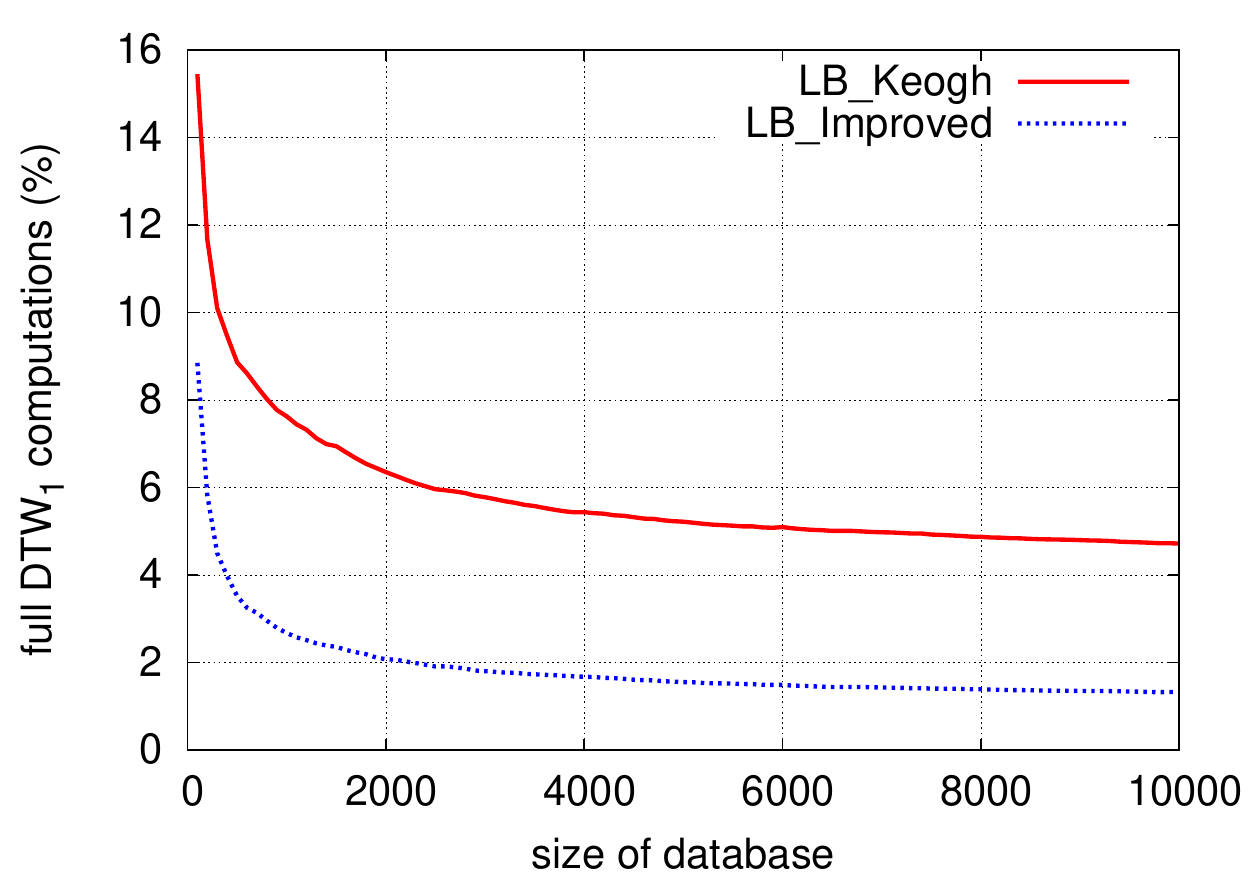}} 
\caption{Nearest-Neighbor Retrieval for the random-walk data set\label{fig:perf3}}
\end{figure*}

\subsection{Shape data sets}

 For the rest of the section, we considered a large collection of time-series derived from shapes~\cite{keogh2006lse,keoghshapedata}.
The first data set is made of heterogeneous shapes which resulted in   5~844 1~024-sample times series.
The second data set is an arrow-head data set with  of 15~000 251-sample time series.
We shuffled randomly each data set so that candidates appear in random order.  We extracted 50~time series 
from each data set, and we present the average nearest-neighbor retrieval times and pruning power
as we consider various fractions of each database (see Fig.~\ref{fig:perf4} and \ref{fig:perf5}).
The results are similar:  LB\_Improved has twice the pruning power and is
faster by a factor of 3.

\begin{figure*}
\centering
  \subfloat[Average Retrieval Time]{\includegraphics[width=0.45\textwidth]{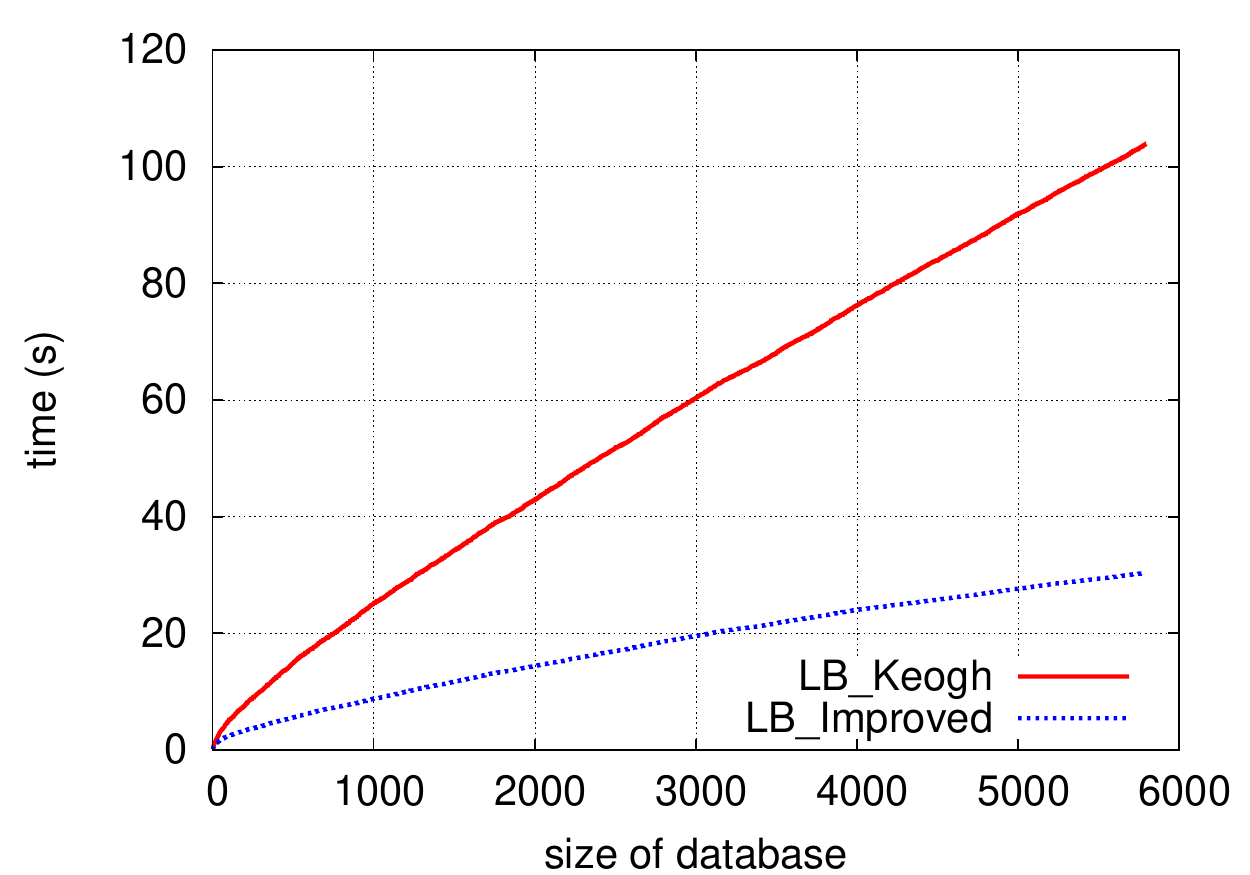}}
  \subfloat[Pruning Power]{\includegraphics[width=0.45\textwidth]{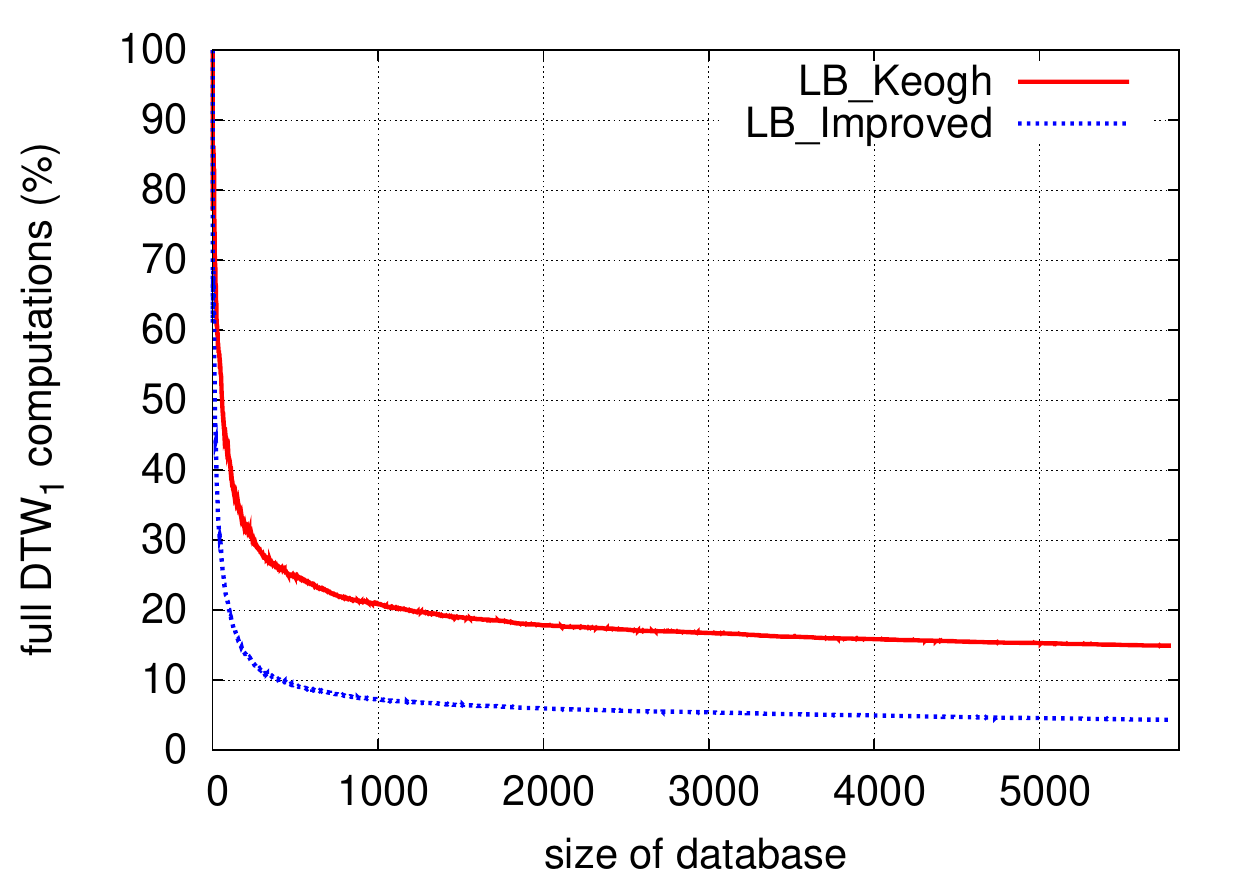}}
\caption{Nearest-Neighbor Retrieval for the heterogeneous shape data set\label{fig:perf4}}
\end{figure*}

\begin{figure*}
\centering
  \subfloat[Average Retrieval Time]{\includegraphics[width=0.45\textwidth]{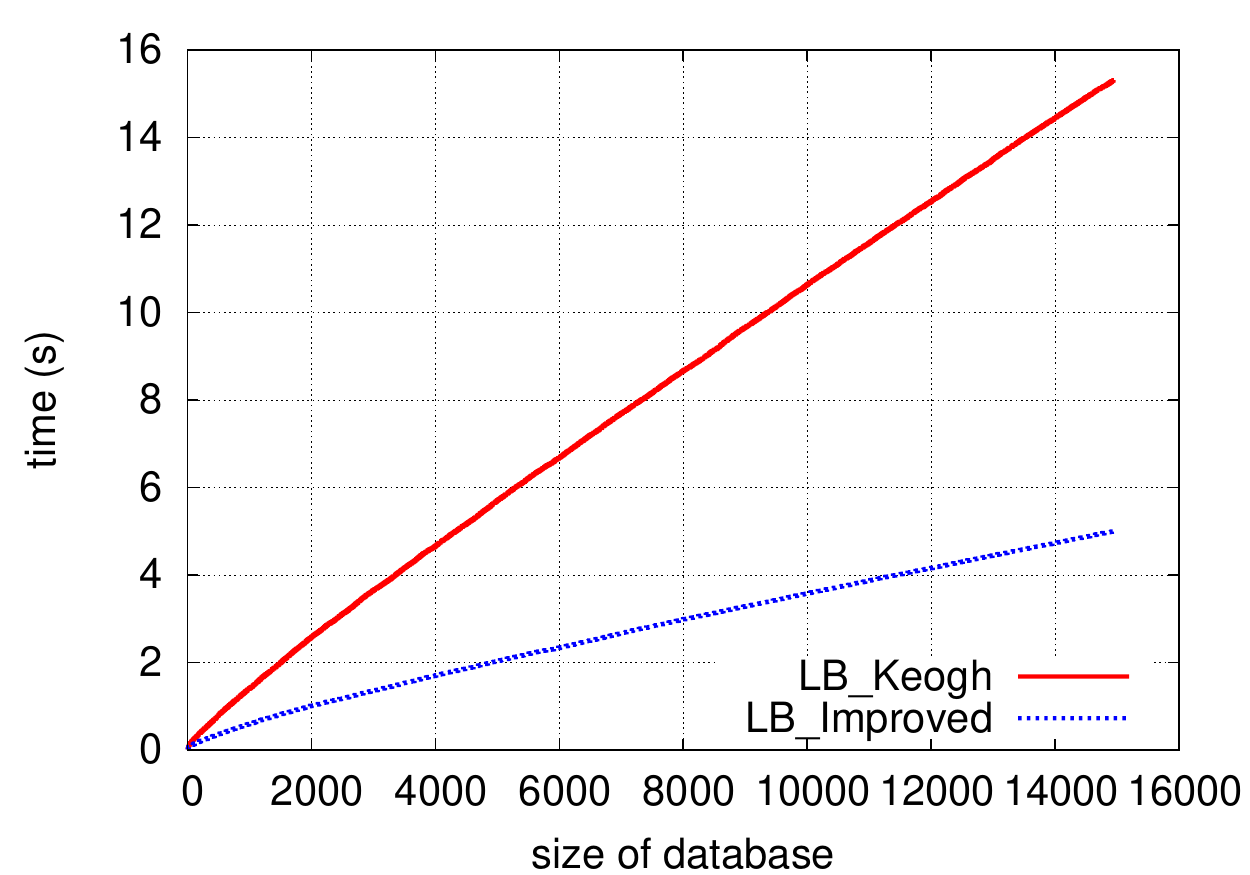}}
  \subfloat[Pruning Power]{\includegraphics[width=0.45\textwidth]{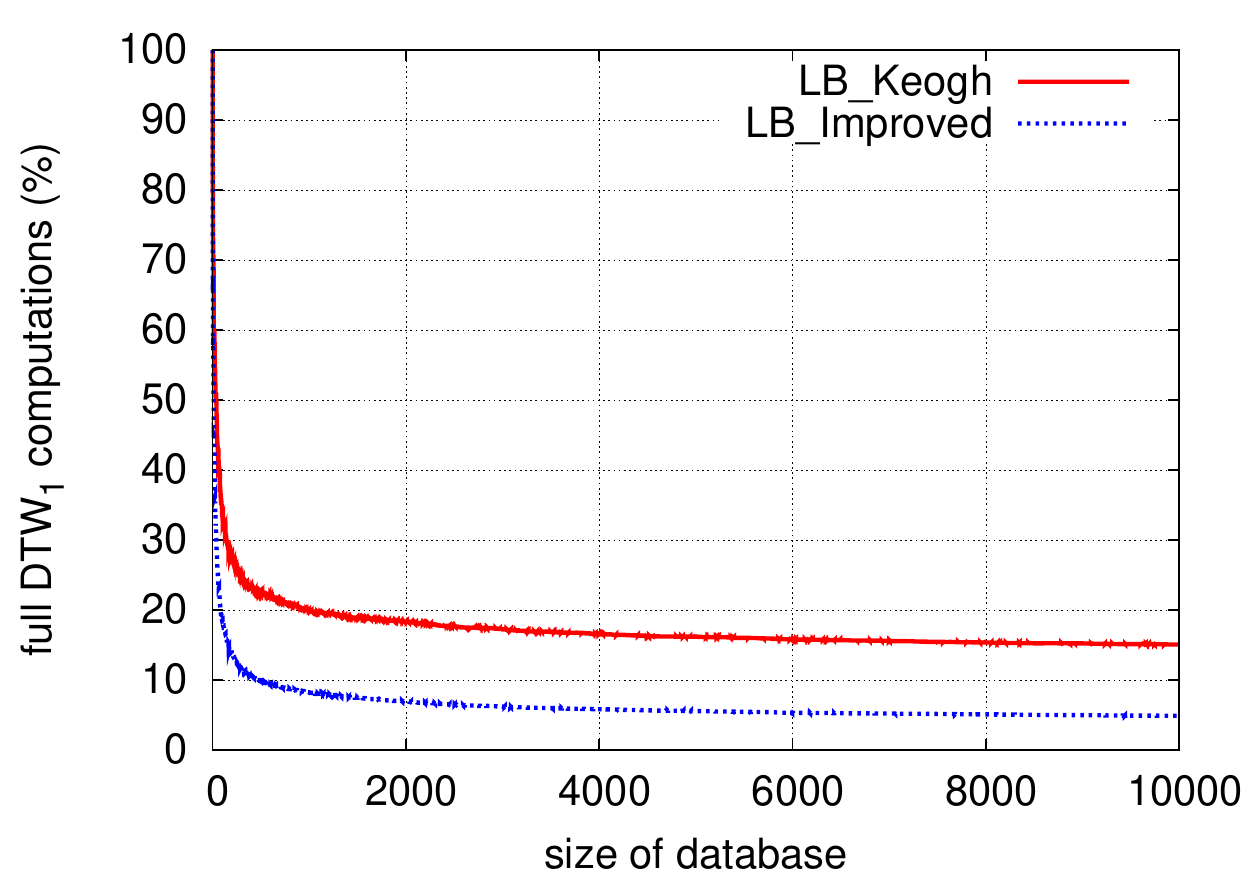}}  
\caption{Nearest-Neighbor Retrieval for the arrow-head shape data set\label{fig:perf5}}
\end{figure*}

 \section{Conclusion}
 
 We have shown that a two-pass pruning technique can improve the retrieval speed
 by up to three times in several time-series databases. We do not use more memory. 
 
 We expect to be able to significantly accelerate the retrieval with
 parallelization. Several instances of Algo.~\ref{algo:lbimproved} can
 run in parallel as long as they can communicate the distance between
 the time series and the best candidate.

 \section*{Acknowledgements}
 The author is supported by NSERC 
grant 261437 and FQRNT grant 112381.

\singlespacing

\bibliographystyle{elsart-num}
\bibliography{../maxminalgo,../lbkeogh} 

\end{document}